\documentclass[11pt,english]{article}
\usepackage[T1]{fontenc}
\usepackage[utf8]{luainputenc}
\usepackage{parskip}
\usepackage{babel}
\usepackage{bm}
\usepackage{amsmath}
\usepackage{amsthm}
\usepackage{amssymb}
\usepackage{graphicx}
\usepackage{geometry}
\usepackage{setspace}
\usepackage[authoryear]{natbib}
\usepackage[pdfusetitle,
 bookmarks=true,bookmarksnumbered=false,bookmarksopen=false,
 breaklinks=false,pdfborder={0 0 1},backref=false,colorlinks=false]
 {hyperref}

\makeatletter
\newenvironment{lyxcode}
	{\par\begin{list}{}{
		\setlength{\rightmargin}{\leftmargin}
		\setlength{\listparindent}{0pt}
		\raggedright
		\setlength{\itemsep}{0pt}
		\setlength{\parsep}{0pt}
		\normalfont\ttfamily}%
	 \item[]}
	{\end{list}}
\theoremstyle{definition}
 \newtheorem{example}{\protect\examplename}
\theoremstyle{definition}
\newtheorem{defn}{\protect\definitionname}
\theoremstyle{plain}
\newtheorem{assumption}{\protect\assumptionname}
\theoremstyle{plain}
\newtheorem{thm}{\protect\theoremname}
\theoremstyle{plain}
\newtheorem{lem}{\protect\lemmaname}
\theoremstyle{plain}
\newtheorem{cor}{\protect\corollaryname}
\theoremstyle{plain}
\newtheorem{prop}{\protect\propositionname}
\theoremstyle{remark}
\newtheorem{rem}{\protect\remarkname}
\theoremstyle{remark}
\newtheorem{claim}{\protect\claimname}

\usepackage{booktabs}
\usepackage{float}
\usepackage{multirow, tabularx}

\usepackage{color}  
\definecolor{ucb}{RGB}{0, 50, 98}
\definecolor{booth}{RGB}{139, 0, 0}
\hypersetup{colorlinks=true,citecolor=booth,urlcolor=booth,anchorcolor=booth,linkcolor=booth}

\usepackage[toc,title,page]{appendix}

\makeatother

\providecommand{\assumptionname}{Assumption}
\providecommand{\claimname}{Claim}
\providecommand{\corollaryname}{Corollary}
\providecommand{\definitionname}{Definition}
\providecommand{\examplename}{Example}
\providecommand{\lemmaname}{Lemma}
\providecommand{\propositionname}{Proposition}
\providecommand{\remarkname}{Remark}
\providecommand{\theoremname}{Theorem}

\begin{document}
\title{Expert Incentives under Partially Contractible States}
\date{\today}
\author{Zizhe Xia\thanks{\protect\href{http://mailto:zizhe-xia@chicagobooth.edu}{zizhe-xia@chicagobooth.edu}
I am indebted to Doron Ravid for his invaluable guidance and advices.
I thank Eric Budish, Kailin Chen, Alex Frankel, Emir Kamenica, Andrew
McClellan, Daniel Rappoport, Lars Stole, Kun Zhang for their valuable
comments and suggestions. All errors are my own.}}\maketitle

\begin{abstract}
A principal hires an agent to acquire costly information about the
state, but it is not possible to pay the agent based on the realized
states. Instead, the principal has access to an  experiment about
the state, and can pay bonuses based on its realization. The agent
is risk neutral and protected by limited liability. I  characterize
what the principal can incentivize the agent to learn, and how to
design contracts to minimize the costs to provide such incentives.
I then study which experiment is always better at incentive provision.
This gives rise to a novel order on information. In the binary-binary
case, this order is characterized by differences in the likelihood
ratios of the two realizations.

\end{abstract}

\thispagestyle{empty}

\newpage{}
\begin{lyxcode}
\setcounter{page}{1}
\end{lyxcode}

\section{Introduction}

A central question in principal-agent models is how information
shapes the payoffs. In a typical moral hazard problem, a principal
hires an agent to produce certain outcomes. The agent exerts some
(typically unidimensional) effort that stochastically affects the
outcomes. The principal does not observed the effort, and can only
contract with the agent based on some noisy information about the
effort. The comparisons of information have been widely studied in
the production moral hazard setting. This literature builds on the
informativeness principle \citet{holmstrom1979moral}, and has been
developed further in various directions.\footnote{Among those, \citet{gjesdal1982information} first show that better
information in moral hazard models does not have to be Blackwell more
informative. \citet{kim1995efficiency} provides a characterization
based on mean preserving spread. \citet{chen2025experiments} further
generalizes this to the linear Blackwell order. \citet{xia2025comparisons}
implements a geometric perspective based on the state dependent utilities
the principal can generate and studies the column space, the conic
span, and the zonotope orders, each corresponding to the comparisons
of information in a different class of moral hazard problems. } 

In many economic applications, however, the agent is hired not to
produce outcomes directly, but to act as an expert who acquires costly
information about some unknown states of the world. This is fundamentally
different from production because information never changes the marginal
distribution of states, while production does change the marginal
distribution of outcomes. As an example, a firm hires a consultant
to predict the long-term profitability of a new product. The consultant
can choose to acquire any information at some cost. The underlying
states are the product's profitability, but they are not realized
for years and cannot be contracted upon. The firm has to collect
noisy information about the states to contract with the consultant.
This can include early data on sales, growth, market penetration,
and consumer satisfaction after the product launch. Despite the prevalence
of these situations, several important economic questions remain unexplored.\footnote{Several papers have considered the problem of incentivizing information
acquisition under different contractibility assumptions. \citet{rappoport2017incentivizing}
and \citet{yoder2022designing} consider contracting on the realizations
of the agent's experiment. \citet{whitmeyer2022buying}, \citet{li2022optimization},
\citet{bloedel2023costly}, and \citet{sharma2024procuring} consider
contracting on the underlying states. However, none of them studies
the comparisons of the principal's contractible information, which
is the main focus of this paper.} To what extent does the noise limit the principal's ability to provide
incentives? What contractible information is better for the principal?
How does the answer compare to the production moral hazard models? 

This paper studies the comparisons of contractible information in
moral hazard settings for information acquisition. I analyze whether
and which expert incentives can be provided at what cost when the
states of the world are non-contractible, but there is some noisy
observation about the states that can be contracted upon. A risk
neutral principal hires an agent to acquire some costly information
about the states, but it is not possible to pay the agent based on
the realized states. Instead, the principal has access to some noisy
(Blackwell) experiment about the states, and can pay bonuses to the
agent based on its realization. The agent is risk neutral and is protected
by limited liability. He can flexibly acquire a conditionally independent
experiment about the states at some cost, but what he learns is neither
observable nor contractible. The agent's information cost is smooth
and Blackwell monotone. The principal has to design a contract to
provide the correct incentives for the agent to acquire the desired
information and report honestly its realization. 

I  characterize what incentives can be provided given the noisy contractible
experiment in Theorem \ref{thm:imp}. Roughly speaking, the agent
can be incentivized to acquire some information if and only if the
principal's information allows her to provide the correct marginal
benefits to the agent. This implementability condition is summarized
in terms of the column space of the principal's experiment.\footnote{A finite Blackwell experiment can be represented as a matrix. Each
row represents a state, each column represents a realization, and
each element specifies the conditional probability of sending a realization
in a state. The column space of the principal's experiment is therefore
the column space of the matrix that represents it. The column space
of a matrix $A\in\mathbb{R}^{N\times M}$ is defined as $\operatorname{Col}A=\{Av:v\in\mathbb{R}^{M}\}$.
It is the set of all linear combinations of its columns.} I also characterize the set of contracts that provides the correct
incentives in Theorem \ref{thm:imp-contract-decomposition}. All such
contracts can be decomposed into three parts, the first part provides
the correct incentives, the second part consists of bonuses that depend
only on the realization of the principal's experiment, and the last
part is a collection of side bets with zero expected value conditional
on each state. The last two parts do not affect the agent's learning
incentives, but play an important role in the principal's cost minimization
problem. 

An immediate consequence of Theorem \ref{thm:imp} is that, when the
principal's experiment has full row rank, she can incentivize any
feasible learning even if her information is arbitrarily noisy (Corollary
\ref{cor:full-imp}).\footnote{If the principal's information is arbitrarily noisy, the cost of
implementation can be arbitrarily high.} The full row rank condition is similar to the identifiability condition
in repeated games \citep{fudenberg1994folk}. It says that the principal
can identify any state distribution if she can observe infinitely
many draws from her experiment. This gives her the ability to create
any state-dependent utilities for the agent. The condition is related
to, though distinct from, the full rank condition in \citet{cremer1988full}.
They prove that an auctioneer can extract all surplus in dominant
strategies using side bets, provided that bidders have linearly independent
interim beliefs. Their full rank condition is similar to mine in the
sense that both allow the principal to construct any state-dependent
utility. Yet, neither condition implies the other, and full surplus
extraction is not possible here because of limited liability.

The analysis also motivates comparing experiments based on their column
spaces. The column space of an experiment is the set of all state
dependent utilities the principal can generate. Say that an experiment
dominates another in the column space order if the column space of
the former contains that of the latter.\footnote{This order is first introduced by \citet{azrieli2022elicitability}
under a different name in the context of elicitation, and later studied
by \citet{xia2025comparisons} in moral hazard problems, who refers
to it as the column space order.} Proposition \ref{prop:imp-comparison} says that an experiment has
a larger column space than another if and only if it can always implement
a larger set of learning, regardless of the agent's cost function.

Next, I study the principal's cost minimization problem, where the
principal seeks the cheapest contract to incentivize the agent to
learn a given experiment -- this minimum cost defines the principal's
indirect cost of information. When the principal's experiment has
full row ranks, I characterize this indirect cost function and the
associated cost-minimizing contracts in Proposition \ref{prop:LL-opt-contract}.
This indirect cost function is neither posterior separable nor Blackwell
monotone. When the principal's experiment has deficient row ranks,
I provide simple linear programs to determine the optimal contracts.

I then characterize what contractible information always has a smaller
indirect cost of information, regardless of the information the agent
is asked to acquire and the cost function the agent has. This defines
a novel indirect cost order that is different from other information
orders in the literature. 

I fully characterize the indirect cost order in the binary-binary
case in Proposition \ref{prop:binary-LL-characterization}.\footnote{That is, the state space is binary, and focus on contractible experiments
with only two realizations.} In this case, an experiment dominates the other in the indirect cost
order if and only if it has larger differences between the likelihood
ratios of the two realizations. This difference can be interpreted
as the marginal productivity of incentives from a contractible experiment.
Incentives to acquire information only depend on how much the agent
gains from making the correct report relative to no information. That
is, only the differences, rather than the levels, of the agent's state
dependent utilities following different reports matter. The likelihood
ratios capture how efficient each realization can be used to generate
state dependent utilities for the agent. The difference between them
then captures how efficient the experiment can be used to generate
information acquisition incentives, which depends on the difference
in the state dependent utilities following different report. 

More generally, a sufficient but not necessary condition for indirect
cost dominance is that an experiment dominates another in the conic
span order (Proposition \ref{prop:LL-cost-comparison}). The conic
span of an experiment is the set of all state dependent utilities
the principal can generate using non-negative payments.\footnote{Mathematically, the conic span of a matrix $A\in\mathbb{R}^{N\times M}$
is defined as $\operatorname{Cone}A=\{Av:v\in\mathbb{R}_{+}^{M}\}$.
It is the set of all non-negative linear combinations of its columns.} Say that an experiment dominates another in the conic span order
if the conic span of the former contains that of the latter. Proposition
\ref{prop:LL-cost-comparison} says that if an experiment has a larger
conic span than another, then it can implement any experiment at a
lower cost.

The remainder of the paper is organized as follows. Section \ref{subsec:lit-review}
reviews the related literature. Section \ref{subsec:example} walks
through a simple example of two Blackwell non-comparable experiments
where one is always better at incentive provision than the other.
Section \ref{sec:model} describes the information acquisition model.
Section \ref{sec:implementability} presents the results of the implementability
problem and comparisons by the implementable sets. Section \ref{sec:IR-cost-min}
solves the cost minimization problem and compares experiments based
on the implementation costs. Section \ref{sec:conclusion} concludes.

\subsection{Related Literature\label{subsec:lit-review}}

This paper studies a problem of information acquisition and elicitation.
A number of papers have also studied the incentive problem in information
acquisition under different assumptions of contractibility. Some papers
assume the realization of the expert's experiment is verifiable and
contractible \citep{rappoport2017incentivizing,yoder2022designing}.
They show that there is no loss to the principal from contracting
on realizations rather than the experiment itself under risk neutrality.
Other papers assume that the expert's learning is not observed, but
the underlying states are contractible \citep{whitmeyer2022buying,bloedel2023costly,muller2024rational,sharma2024procuring}.
\citet{bloedel2023costly} study the behavior of this problem when
the number of agents tends to infinity. \citet{azrieli2021monitoring}
studies a similar problem with non-contractible states but the principal
can hire multiple agents and compare their reports.

Among those, the closest paper is \citet{sharma2024procuring}. They
study a similar problem where a principal hires an expert with posterior
separable costs to acquire information, but the states are contractible.
They find that the principal can incentivize the expert to acquire
any feasible information. I take one step further and show that the
full implementability holds as long as the contractible information
reveals even a little bit about every state. Moreover, my framework
allows for comparing the contractible information, and shreds light
on what information is better at incentive provision. I also extend
this framework beyond posterior separable cost functions. 

The literature on proper scoring rules also studies such contracting
problems with a focus on information elicitation (for example, \citealp{glenn1950verification},
\citealp{mccarthy1956measures}, \citealp{savage1971elicitation},
and more recent works by \citealp{chen2021optimal} and \citealp{li2022optimization}).
This line of work explores how to elicit an expert's true belief about
a probabilistic state by designing rewards, or ``scores'', that
depends on both the expert's reported belief and the realized state.
There, truthtelling is guaranteed if the expert's value function is
convex in the reported belief. \citet{liu2023surrogate} and \citet{azrieli2022elicitability}
consider settings where the principal can only contract on a noisy
experiment about the states. \citet{liu2023surrogate} restrict attention
to binary states, and their focus is how to cross validate reports
from multiple experts. \citet{azrieli2022elicitability} characterizes
the contract for eliciting the expert's belief given the principal's
noisy experiment. They also develop an order that compares which experiment
allows the principal to elicit a larger set of beliefs. This order
coincides with the column space order, which also characterizes the
comparisons of implementability studied in this paper. The same order
applies because in both problems incentives are governed by the agent’s
state dependent utility. The column space precisely characterizes
which such utilities can be implemented.

This paper differs from the scoring rule papers mentioned above in
that the expert's information is acquired endogenously and flexibly.
In addition to elicitation, the contract also has to provide the correct
learning incentives.\footnote{\citet{li2022optimization} also concerns the expert's incentives
to acquire information. However, the expert in their model only has
a binary choice of whether to acquire information. The expert in my
model can flexibly acquire any information. He can choose not only
how much to learn, but also what to learn about. } Moreover, in my model, information elicitation does not come with
additional costs if correct incentives are provided for information
acquisition due to the convexity of the agent's value function. 

At a higher level, this paper also belongs to the literature on the
comparison of information, pioneered by Blackwell \citeyearpar{blackwell1951comparison,blackwell1953equivalent}.
In Blackwell's papers, he compares information across all decision
problems and characterizes the Blackwell order using garbling. The
Blackwell order turns out to be too restrictive because it has many
non-comparabilities. A large literature is devoted to comparison of
information by restricting attention to decision problems that are
monotone \citep{lehmann2011comparing,kim2023comparing}, satisfy a
single crossing property \citep{persico2000information}, have the
interval dominance property \citep{quah2009comparative}.\footnote{Other papers take a route different from restricting the set of decision
problems. For example, \citet{brooks2024comparisons} consider robustness
to additional information. \citet{mu2021blackwell} study multiple
independent draws from the same experiment.} Focusing on production moral hazard problems, \citet{holmstrom1979moral}
proposes the informativeness principle as one way to compare information,
which is then extended by \citet{kim1995efficiency}, \citet{chen2025experiments}
and \citet{xia2025comparisons}. On the other hand, my paper compares
information in moral hazard problems where the agent's task is acquiring
information. 

Among those, the closest is \citet{xia2025comparisons} who uses a
geometric approach to compare experiments in a flexible production
model based on the agent's state dependent utilities that can be generated.
He finds that the column space order characterizes the comparison
of implementability, and the conic span order characterizes the cost
comparison when the agent is risk neutral and protected by limited
liability. A surprising finding in this paper is that the cost comparison
is different when the agent's task is to acquire information. The
conic span order is sufficient but not necessary. This is because
incentives are provided via different channels in the two models.
When the agent is hired to produce an outcome, what matters for incentives
is his state dependent utility following each outcome. When the agent
is hired to acquire information, any state specific bonus to the agent
provides no incentive as the agent does not control the states. Incentives
to acquire information are provided through making information valuable
by designing the differences in the state dependent utilities following
different reports.

\subsection{An Example \label{subsec:example}}

I now present a simple example in which two experiments are Blackwell
non-comparable, yet one always yields better expert incentives than
the other. I focus on the case of a binary state space and experiments
with binary realizations. This binary-binary setting highlights the
difference between the comparison of information developed in this
paper and those studied in the existing literature.\footnote{The Blackwell order is known to be too strong for comparing information
in moral hazard problems. \citet{chen2025experiments} and \citet{xia2025comparisons}
propose the zonotope order as the correct criterion in production
moral hazard settings when the agent can be arbitrarily risk averse.
\citet{xia2025comparisons} further shows that, with a risk-neutral
agent under limited liability, the conic span order is the correct
order to compare information. In the binary-binary case, however,
these orders coincide with the Blackwell order. This example therefore
illustrates that when the agent's task is to acquire information,
the relevant incentive problem is different from production, and the
information order is also different.} 

The setup is as follows. A principal hires an agent to acquire some
information about some unknown binary state $\omega\in\Omega=\{\omega_{1},\omega_{2}\}$.
For example, a firm hires a consultant to assess the long-term profitability
of some new product. The agent is risk neutral, protected by limited
liability, and can flexibly acquire information at some cost, but
the acquisition process is unobservable to the principal. 

To provide incentives, the principal observes only the outcome of
a noisy experiment about the state,\footnote{The agent's information is assumed to be conditionally independent
of the principal's experiment.} since the state itself is not contractible - it can take years for
the long-term profitability to be realized, and the consultant must
be paid earlier. Examples of such noisy experiments include early
data on sales, growth, or customer satisfaction. The principal asks
the agent to acquire and report information, and pays him based on
his report and the realization of the noisy experiment. Suppose the
principal can choose between two such experiments shown in Example
\ref{exa:EPIR-not-Blackwell}, and wishes to know which makes incentives
cheaper to provide.
\begin{example}
\label{exa:EPIR-not-Blackwell} Consider the following two experiments
in the matrix form with each row representing a state and each column
a realization,
\begin{align*}
\mathcal{E}_{1}=\begin{array}{c}
\omega_{1}\\
\omega_{2}
\end{array}\overset{\begin{array}{cc}
y_{1} & y_{2}\end{array}}{\begin{bmatrix}0.7 & 0.3\\
0.3 & 0.7
\end{bmatrix}}, & \;\mathcal{E}_{2}=\begin{array}{c}
\omega_{1}\\
\omega_{2}
\end{array}\overset{\begin{array}{cc}
y_{1}' & y_{2}'\end{array}}{\begin{bmatrix}0.5 & 0.5\\
0.2 & 0.8
\end{bmatrix}}.
\end{align*}
Under the uniform prior $\mu_{0}:=\Pr(\omega_{2})=0.5$, $\mathcal{E}_{1}$
induces posteriors $y_{1}=0.3$ and $y_{2}=0.7$, and $\mathcal{E}_{2}$
induces posteriors $y_{1}'=0.29$ and $y_{2}'=0.61$.

\end{example}
The experiments $\mathcal{E}_{1}$ and $\mathcal{E}_{2}$ are clearly
Blackwell non-comparable. Yet, when used to provide expert incentives,
$\mathcal{E}_{1}$ is always better than $\mathcal{E}_{2}$, regardless
of the agent's information cost or the specific experiment the principal
wishes the agent to learn. 

The remainder of this section develops the intuition behind the example.
For simplicity, I restrict attention to the simple case where the
principal wants to incentivize to agent to acquire some binary experiment
$\mathcal{E}_{A}$ that reveals whether the state is more likely to
be $\omega_{1}$ or $\omega_{2}$. The result holds for arbitrary
$\mathcal{E}_{A}$. In this setting, the principal asks the agent
to report $x_{1}$ or $x_{2}$, which stand for state $\omega_{1}$
or $\omega_{2}$ being more likely, and compensates him based on the
report and the realization of the her own experiment.

Under limited liability, providing incentives at minimum cost reduces
to considering pay-if-correct contracts: The agent is paid only if
his report and the principal's experiment realization point in the
same direction. Such contracts leave no excessive rents on the table
and are optimal. Suppose, for now, that the principal contracts using
$\mathcal{E}_{1}$. A pay-if-correct contract pays the agent $t_{1}$
if he reports $\omega_{1}$ as more likely and $y_{1}$ is realized,
$t_{2}$ if he reports $\omega_{2}$ as more likely and $y_{2}$ is
realized, and zero otherwise.

Next, I turn to what matters for providing incentives. In the context
of information acquisition, incentives are created by making information
more valuable to the agent. The value of information is captured by
the agent's value function, which maps his posterior beliefs about
the state to the maximum expected payment attainable under the contract. 

To illustrate, Figure \ref{fig:value-func} Panel A plots the agent's
expected payment as a function of his belief about the realization
of the principal's experiment $\mathcal{E}_{1}$, depending on his
report. The agent's value function corresponds to only a subset of
this plot. The reason is that, even if the agent learns with certainty
that the state is $\omega_{1}$ or $\omega_{2}$, the probabilities
of $y_{1}$ or $y_{2}$ remain less than one because $\mathcal{E}_{1}$
is noisy. Panel B plots the expected payment as a function of the
agent's belief about the state on the horizontal axis, which is a
re-scaling of Panel A.\footnote{More specifically, Panel B is a re-scaling of Panel A, restricted
to the range of probabilities consistent with a feasible posterior
over the state.} The agent's value function is simply the upper envelope of the two
lines in Panel B. 

\begin{figure}[th]
\begin{centering}
\caption{Incentives for Information Acquisition \label{fig:value-func}}
\par\end{centering}
\begin{centering}
\medskip{}
\par\end{centering}
\begin{centering}
\begin{minipage}[t]{0.48\columnwidth}%
\begin{center}
{\small Panel A: Expected Payment from a Contract}{\small\par}
\par\end{center}
\begin{center}
\includegraphics[width=1\columnwidth]{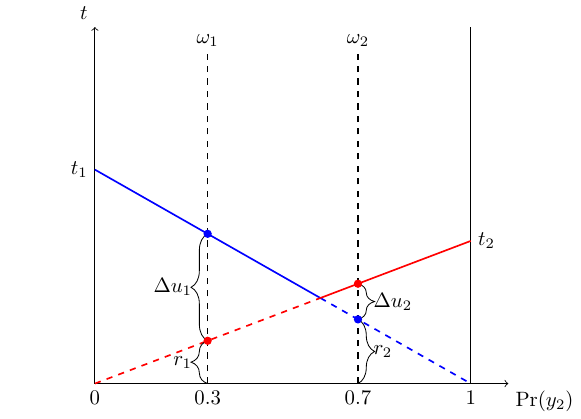}
\par\end{center}%
\end{minipage}\quad{}%
\begin{minipage}[t]{0.48\columnwidth}%
\begin{center}
{\small Panel B: Value Function}{\small\par}
\par\end{center}
\begin{center}
\includegraphics[width=1\columnwidth]{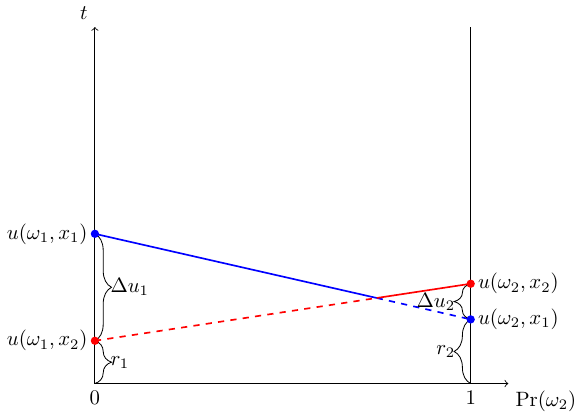}
\par\end{center}%
\end{minipage}
\par\end{centering}
\centering{}%
\begin{minipage}[t]{0.9\columnwidth}%
\begin{singlespace}
{\scriptsize Notes: This figure depicts the agent's expected payment
and value function from a pay-if-correct contract. Panel A plots the
agent's expected payment from a pay-if-correct contract as a function
of his belief about the realization of the principal's experiment,
$\Pr(y_{2})$. The blue line is the agent's expected payment when
he reports $x_{1}$ ($\omega_{1}$ is more likely), and the red line
is the agent's expected payment when he reports $x_{2}$ ($\omega_{2}$
is more likely). The agent always reports the $x$ that maximizes
his expected payment. The two dashed vertical lines are the agent's
belief about the realization $\Pr(y_{2})$ when he is certain about
the state. Since the principal's experiment is noisy, $\Pr(y_{2})$
has to reside in $[0.3,0.7]$. Panel B plots the agent's value function.
This is his expected payment from the contract as a function of his
belief about the state, $\Pr(\omega_{2})$. This is simply Panel A
zoomed into $[0.3,0.7]$ with the horizontal axis redefined. }{\scriptsize\par}
\end{singlespace}

\end{minipage}
\end{figure}

I now read off the agent's value function his information acquisition
incentives and rents. Formally, let $u(\omega,x)$ be the agent's
expected payment when he reports $x$ and the realized state is $\omega$.\footnote{Given the state $\omega$, the principal's experiment pins down a
distribution over $y_{1}$ and $y_{2}$. The term $u(\omega,x)$ is
then the expected payment with respect to this distribution.} Incentives to acquire information depend only on how much the agent
gains from learning the true state relative to having no information.
That is, they are determined by the differences in the agent's state
dependent utility, namely, 
\[
\Delta u_{1}:=u(\omega_{1},x_{1})-u(\omega_{1},x_{2}),\:\Delta u_{2}:=u(\omega_{2},x_{2})-u(\omega_{2},x_{1}).
\]
Here, $\Delta u_{1}$ ($\Delta u_{2}$) is measures the additional
payment the agent can obtain when the state is $\omega_{1}$ ($\omega_{2}$)
and he makes the correct report, relative to the wrong report. Call
$(\Delta u_{1},\Delta u_{2})$ an incentive profile. 

By contrast, any payment that depends solely on the realized state
- independent of the agent's report - provides no incentives. In this
case, 
\[
r_{1}:=u(\omega_{1},x_{2}),\;r_{2}:=u(\omega_{2},x_{1}),
\]
are rents to the agent. They are payments the agent can guarantee
even without knowing the state. In Figure \ref{fig:value-func}, incentives
and rents are also labeled accordingly.

\begin{figure}[th]
\begin{centering}
\caption{Optimal Pay-if-Correct Contracts \label{fig:opt-contracts}}
\par\end{centering}
\begin{centering}
\medskip{}
\par\end{centering}
\begin{centering}
\begin{minipage}[t]{0.48\columnwidth}%
\begin{center}
{\small Panel A: Provide $(\Delta u_{1},0)$ with $\mathcal{E}_{1}$}{\small\par}
\par\end{center}
\begin{center}
\includegraphics[width=1\columnwidth]{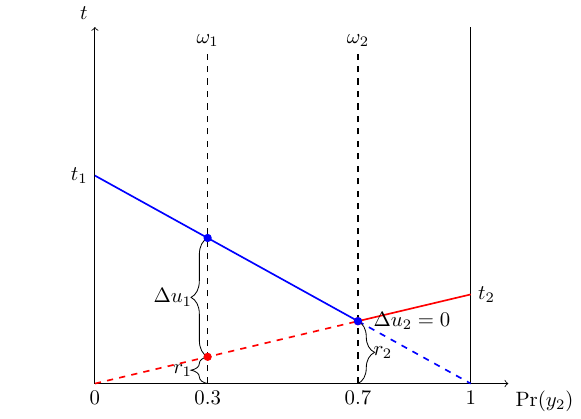}
\par\end{center}%
\end{minipage}\quad{}%
\begin{minipage}[t]{0.48\columnwidth}%
\begin{center}
{\small Panel B: Provide} $(0,\Delta u_{2})${\small{} with $\mathcal{E}_{1}$}{\small\par}
\par\end{center}
\begin{center}
\includegraphics[width=1\columnwidth]{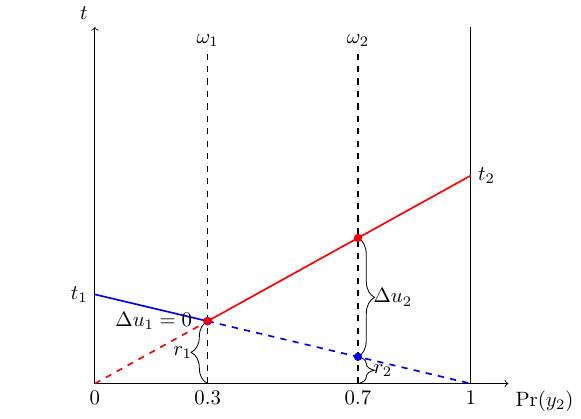}
\par\end{center}%
\end{minipage}
\par\end{centering}
\begin{centering}
\bigskip{}
\par\end{centering}
\begin{centering}
\begin{minipage}[t]{0.48\columnwidth}%
\begin{center}
{\small Panel A: Provide $(\Delta u_{1},0)$ with $\mathcal{E}_{2}$}{\small\par}
\par\end{center}
\begin{center}
\includegraphics[width=1\columnwidth]{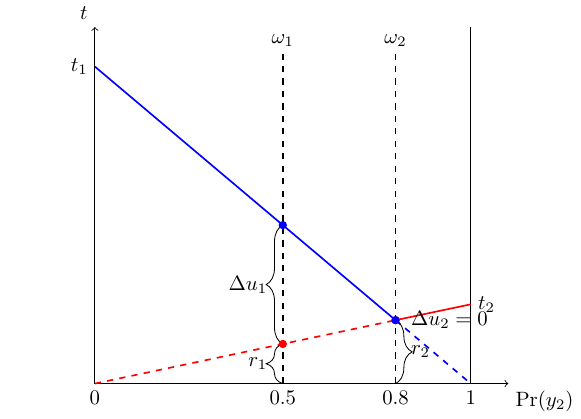}
\par\end{center}%
\end{minipage}\quad{}%
\begin{minipage}[t]{0.48\columnwidth}%
\begin{center}
{\small Panel B: Provide} $(0,\Delta u_{2})${\small{} with $\mathcal{E}_{2}$}{\small\par}
\par\end{center}
\begin{center}
\includegraphics[width=1\columnwidth]{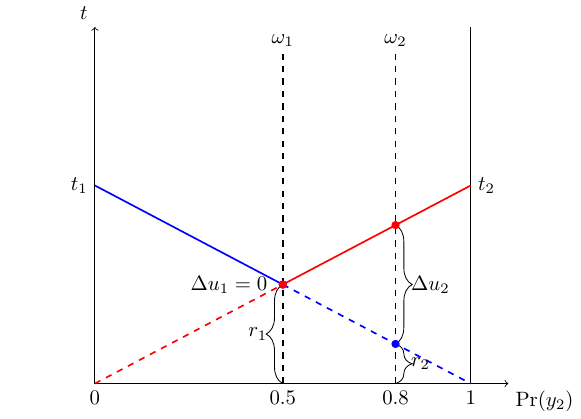}
\par\end{center}%
\end{minipage}
\par\end{centering}
\centering{}%
\begin{minipage}[t]{0.9\columnwidth}%
\begin{singlespace}
{\scriptsize Notes: This figure depicts the optimal contracts to provide
certain incentive profiles. Panels A and B plot the pay-if-correct
contract for $(\Delta u_{1},0)$ and $(0,\Delta u_{2})$ under $\mathcal{E}_{1}$.
Panels C and D plot the pay-if-correct contract for $(\Delta u_{1},0)$
and $(0,\Delta u_{2})$ under $\mathcal{E}_{2}$. Together, they show
that $\mathcal{E}_{1}$ can always provide cheaper expert incentives
than $\mathcal{E}_{2}$, because the rents $r_{1}$ and $r_{2}$ in
Panels A and B are lower than the corresponding rents in Panels C
and D. }{\scriptsize\par}
\end{singlespace}

\end{minipage}
\end{figure}

To show that $\mathcal{E}_{1}$ provides cheaper incentives than $\mathcal{E}_{2}$
in all circumstances, it suffices to verify that, for any incentive
profile $(\Delta u_{1},\Delta u_{2})$, the rents paid to the agent
are lower under $\mathcal{E}_{1}$ than $\mathcal{E}_{2}$. This comparison
does not depend on the agent’s information cost function, since the
principal’s information does not affect the agent’s cost.\footnote{See Section \ref{sec:model} for details on the cost function.} 

In the binary-binary setting, incentives are linear in the pay-if-correct
contracts. This means we can analyze the simpler incentive profiles
$(\Delta u_{1},0)$ and $(0,\Delta u_{2})$ separately. Linearity
then implies that the contract for $(\Delta u_{1},\Delta u_{2})$
is obtained by adding the corresponding pay-if-correct subcontracts,
and the total rents are the sum of the rents from the subcontracts. 

Consider first the contract that provides $(\Delta u_{1},0)$. The
goal is to determine how much rents $r_{1}$ and $r_{2}$ are needed
to generate a unit of $\Delta u_{1}$ while keeping $\Delta u_{2}=0$.
In other words, we are interested in the marginal cost of incentives.
The corresponding contract is shown in Figure \ref{fig:opt-contracts}
Panel A. The relation between $r_{1},r_{2}$ and $\Delta u_{1}$ follows
from similar triangles, 
\[
\dfrac{r_{2}}{r_{1}}=\dfrac{0.7}{0.3},\;\dfrac{r_{2}}{\Delta u_{1}+r_{1}}=\dfrac{0.3}{0.7}.
\]
Solving gives that producing a unit of $\Delta u_{2}$ requires $0.225$
units of $r_{1}$ and $0.525$ units of $r_{2}$. Similarly, we can
read off Panel B of Figure \ref{fig:opt-contracts} to conclude that
producing a unit of $\Delta u_{2}$ requires $0.525$ units of $r_{1}$
and $0.225$ units of $r_{2}$. 

On the other hand, we can do the same calculations for experiment
$\mathcal{E}_{2}$, as shown in Panels C and D. Every unit of $\Delta u_{1}$
requires $0.33$ units of $r_{1}$ and $0.53$ units of $r_{2}$.
Every unit of $\Delta u_{2}$ requires $0.83$ units of $r_{1}$ and
$0.33$ units of $r_{2}$. 

Taking together, the rents required to provide a unit of incentives
is always cheaper under $\mathcal{E}_{1}$ than $\mathcal{E}_{2}$.
Thus, $\mathcal{E}_{1}$ always outperforms $\mathcal{E}_{2}$ for
expert incentives despite that they are Blackwell non-comparable.
Section \ref{sec:IR-cost-min} generalizes this comparison to all
binary-binary experiments, providing a characterization of expert-incentive
comparisons based on the differences in likelihood ratios.

\section{Model\label{sec:model}}

\subsection{Preliminaries}

I adopt the following notations throughout the paper: Matrices are
denoted by uppercase letters, e.g., $\mathcal{E},T$; vectors are
denoted by lowercase letters, e.g., $x,\mu_{0}$; inequality $T\geq0$
means every entry of $T$ is above zero; $\boldsymbol{1}$ and $\boldsymbol{0}$
are vectors or matrices of ones and zeros of conformable shape, and
the shape will be indicated as a subscript when necessary, e.g., $\boldsymbol{1}_{M}$
or $\boldsymbol{0}_{M\times N}$. 

\subsubsection*{States}

There is a set of $N$ states $\Omega=\{\omega_{n}\}_{n=1}^{N}$ with
an interior prior $\mu_{0}\in\operatorname{int}\Delta\Omega$.\footnote{I use $\operatorname{int}\Delta\Omega$ to denote the interior of
$\Delta\Omega$. This is the set of all interior beliefs.} The prior $\mu_{0}$ is viewed as a vector in $\mathbb{R}^{N}$.
Let $\mu_{0}^{n}$ denote the probability of $\omega_{n}$ at $\mu_{0}$.

\subsubsection*{Information}

Information is modeled as (Blackwell) experiments. I restrict attention
to finite experiments. They can be represented by row stochastic matrices,
where each row represents a state $\omega_{n}\in\Omega$, each column
represents a realization $x$, and each entry is the conditional probability
$\Pr\left(x\mid\omega_{n}\right)$ of realization $x$ in state $\omega_{n}$.
\begin{defn}
[Blackwell Experiments] An experiment $\mathcal{E}$ with $M$ realizations
is an $N\times M$ row stochastic matrix, that is, every entry is
weakly positive $\mathcal{E}\geq0$, and every row of $\mathcal{E}$
sums to one $\mathcal{E}\bm{1}_{M\times1}=\bm{1}_{N\times1}$. Let
$\mathcal{E}(x\mid\omega_{n})$ denote the conditional probability
of realization $x$ in state $\omega_{n}$, and $\mathcal{E}(x)$
be the unconditional probability of realization $x$. Let $E^{M}$
be the set of all experiments with $M$ realizations. Let $E=\cup_{M=1}^{\infty}E^{M}$
be the set of all finite experiments.
\end{defn}
 Let $\left\langle \mathcal{E}\mid\mu_{0}\right\rangle $ denote
the distribution of realizations of $\mathcal{E}$ when the prior
is $\mu_{0}$. When there is no ambiguity, I drop the dependency on
the prior and also use $\mathcal{E}$ to represent this distribution
for brevity. Given the prior $\mu_{0}$, each realization $x$ of
$\mathcal{E}$ induces a posterior belief. Use $x$ to also represent
this posterior belief. $\mathcal{E}$ thus represents a distribution
of posterior beliefs. It is well known from \citet{kamenica2011bayesian}
that a distribution of posterior beliefs can be induced by an experiment
if and only if it is Bayes plausible, that is, it averages back to
the prior. Due to this result, I will use $E^{M}$ to also represent
the set of Bayes plausible posterior distributions supported on $M$
posteriors, and $E$ the set of all Bayes plausible posterior distributions.

Write $\mathcal{E}\geq_{\text{B}}\mathcal{E}'$ if experiment $\mathcal{E}$
Blackwell dominates $\mathcal{E}'$. \citet{blackwell1951comparison}'s
theorem says that $\mathcal{E}\geq_{B}\mathcal{E}'$ if and only if
$\mathcal{E}'$ is a garbling of $\mathcal{E}$, or equivalently,
$\mathcal{E}'=\mathcal{E}G$ for some row stochastic matrix $G$ with
positive entries $G\geq0$ and each row summing to one $G\bm{1}=\bm{1}$.
In this case, $G$ is called a garbling matrix since it garbles the
information in $\mathcal{E}$ to produce $\mathcal{E}'$.

\subsection{Contracting for Information Acquisition}

\subsubsection*{The Principal}

A risk neutral principal (she) hires an agent (he) to acquire some
(exogenously given) experiment $\mathcal{E}_{A}\in E^{K}$ that induces
the set of posteriors $X:=\{x_{k}\}_{k=1}^{K}$.\footnote{For example, the principal needs $\mathcal{E}_{A}$ for some (unmodeled)
decision problem. } The principal does not observe and cannot contract directly on the
agent's learning process. Moreover, the states are not contractible.
Instead, the principal exogenously has access to some experiment $\mathcal{E}_{P}\in E^{M}$
that induces the set of posteriors $Y:=\{y_{m}\}_{m=1}^{M}$.  $\mathcal{E}_{P}$
is publicly known and its realization is contractible. This setting
nests contractible states as a special case when $\mathcal{E}_{P}$
fully reveals the states.

To provide incentives, the principal can offer a contract to the agent
with bonuses determined by both the realization of $\mathcal{E}_{P}$
and what the agent claims to have learned. More specifically, a contract
under $\mathcal{E}_{P}$ is a pair $(\mathcal{M},T)$. $\mathcal{M}$
is a finite space of messages large enough for the agent to report
his findings after learning. For example, $\mathcal{M}$ can be the
set of possible consulting reports the agent may write. $T:Y\times\mathcal{M}\to\mathbb{R}$
specifies the payment to the agent given the realization of $\mathcal{E}_{P}$
and the agent's reported message. 

Without loss, we can always set $\mathcal{M}=X$ and use the payment
rule $T$ to represent the contract. This comes from a revelation-like
argument. That is, the principal asks the agent ``which realization
of $\mathcal{E}_{A}$ did you observe'' and requires the agent's
answer to be a realization of $\mathcal{E}_{A}$, even though the
agent may acquire experiments other than $\mathcal{E}_{A}$. From
now on, I drop the dependency on $\mathcal{M}$ and use $T$ to represent
the contract.
\begin{defn}
[Contracts] A contract under $\mathcal{E}_{P}$ is a function $T:Y\times X\to\mathbb{R}$
where $T(y_{m},x_{k})$ specifies how much the principal has to pay
the agent when the realization of $\mathcal{E}_{P}$ is $y_{m}$ and
the agent reports $x_{k}$. 
\end{defn}
For notational convenience, the contract $T$ will be viewed as an
$M\times K$ matrix of payments with rows representing realizations
of $\mathcal{E}_{P}$ and columns representing the agent's reports.
Let $T_{k}:=T(\cdot,x_{k})$ denote the $k$-th column of this matrix.
$T_{k}$ represents the vector of payments for each realization of
$\mathcal{E}_{P}$ when the agent reports $x_{k}$.

\subsubsection*{The Agent}

The agent is also risk neutral and does not intrinsically care about
the states.\footnote{In reality, the agent may have material preferences about the states
either due to his private benefits or reputational concerns. My model
can be extended to consider the agent's material preferences. } He has an outside option with its payoff normalized to zero. 

The agent can acquire an experiment about the states at a cost. His
information cost is commonly known and is described by some function
$C:\Delta\Delta\Omega\to\mathbb{R}_{+}\cup\{+\infty\}$, where an
infinite cost means the experiment is not feasible. It suffices to
consider the set of feasible experiments. $C$ can take any form as
long as it is smooth and Blackwell monotone, and satisfies some mild
technical assumptions. I will come back to the assumptions on $C$
later in Section \ref{subsec:info-cost}. 

I assume that any experiment the agent may acquire is conditionally
independent with the principal's experiment $\mathcal{E}_{P}$. This
means the agent is not able to learn about the realization of $\mathcal{E}_{P}$
directly. He can only do so indirectly by learning about the states.
This assumption is likely to hold if the principal and the agent do
not know each other's sources of information. 

Given a contract $T$, the agent decides whether to accept it, and
if so, chooses what to learn and what to report to maximize the expected
payment minus the information cost. Formally, once the agent has accepted
the contract, he chooses an experiment $\mathcal{E}$ and a reporting
rule $s:\Delta\Omega\to\Delta X$ to solve
\begin{align*}
\max_{\mathcal{E}\in E,s} & \:\mathbb{E}_{\mu\sim\left\langle \mathcal{E}\mid\mu_{0}\right\rangle ,y_{m}\sim\left\langle \mathcal{E}_{P}\mid\mu\right\rangle }\left[T(y_{m},s(\mu))\right]-C(\mathcal{E}),
\end{align*}
where the first term is the expected payment with $T(y_{m},\cdot)$
extended to include probablistic reports and the expected value taken
over the agent's posterior belief $\mu$ and the distribution of $\mathcal{E}_{P}$'s
realizations conditional on $\mu$, and the second term is the information
cost.\footnote{It is worth mentioning that $\mathcal{E}_{P}$'s realizations are
distributed according to $\left\langle \mathcal{E}_{P}\mid\mu\right\rangle $
rather than $\left\langle \mathcal{E}_{P}\mid\mu_{0}\right\rangle $.
This is because when the agent acquires information and updates his
belief about the states, he also updates his belief about the $\mathcal{E}_{P}$'s
realizations. The agent should evaluate the expected payment using
the updated belief. } When indifferent, assume the agent always breaks ties in favor of
the principal.

The problem can be simplified as the reporting strategy $s$ does
not play a role. This again comes from a revelation-like argument.
It is without loss to focus on direct recommendation experiments $\mathcal{E}\in E^{K}$
with $K$ realizations, each corresponding to recommending a report
$x_{k}\in X$, together with a reporting rule $s^{*}$ that always
follows the recommendation. Using a direct recommendation experiment,
the agent's problem can be simplified as 
\begin{equation}
\max_{\mathcal{E}\in E^{K}}\;\mathbb{E}_{\mu_{k}\sim\mathcal{E}}\left[\mu_{k}\cdot\mathcal{E}_{P}T_{k}\right]-C(\mathcal{E}),\label{eq:A-problem}
\end{equation}
where $\mu_{k}$ is the posterior belief induced by $\mathcal{E}$
at which the agent reports the $k$-th message $x_{k}$; $\mathcal{E}_{P}T_{k}$
is the product of the matrix $\mathcal{E}_{P}$ and the vector $T_{k}$,
and this product represents a vector of the agent's expected payment
every state when he reports $x_{k}$; the dot product $\mu_{k}\cdot\mathcal{E}_{P}T_{k}$
is the agent's interim expected payment when he holds belief $\mu_{k}$
and reports the $k$-th message $x_{k}$. 

The agent's problem (\ref{eq:A-problem}) is therefore a standard
problem of rational inattention. $\mathcal{E}_{P}T$ specifies the
agent's state-dependent utility for every pair of report and realized
state, and the agent optimally acquires information about the state
to decide on which $x_{k}$ to report. 

\subsubsection*{Timing}

The timing of the model is as follows. Nature draws $\omega$ which
is not observed by the principal or the agent. The principal offers
a contract $T$ to the agent before learning the realization of $\mathcal{E}_{P}$.
The agent chooses whether to accept the contract, and if so, what
experiment $\mathcal{E}$ to learn. Nature draws a realization of
$\mathcal{E}$ given $\omega$, and the agent chooses what to report
given this realization. Nature draws a realization of $\mathcal{E}_{P}$
conditionally independent of $\mathcal{E}$ given $\omega$, and the
principal pays the agent based on his report and the realization of
$\mathcal{E}_{P}$. 

\subsubsection*{The Contracting Problem}

The principal's objective is minimize the expected cost of contract
$T$ subject to the agent's incentive constraint (IC), participation
constraint (PC), and limited liability (LL). Formally, the principal
solves
\begin{align}
\kappa^{C,\mu_{0}}(\mathcal{E}_{A},\mathcal{E}_{P}):= & \min_{t\in\mathbb{R}^{M\times K}}\;\mathbb{E}_{x_{k}\sim\left\langle \mathcal{E}_{A}\mid\mu_{0}\right\rangle ,y_{m}\sim\left\langle \mathcal{E}_{P}\mid x_{k}\right\rangle }\left[T(y_{m},x_{k})\right],\label{eq:P-problem}\\
\text{s.t.(IC) } & \mathcal{E}_{A}\in\underset{\mathcal{E}\in E^{K}}{\operatorname{argmax}}\:\mathbb{E}_{\mu_{k}\sim\mathcal{E}}\left[\mu_{k}\cdot\mathcal{E}_{P}T_{k}\right]-C(\mathcal{E}),\label{eq:agent-IC}\\
\text{(PC) } & \max_{\mathcal{E}\in E^{K}}\;\mathbb{E}_{\mu_{k}\sim\mathcal{E}}\left[\mu_{k}\cdot\mathcal{E}_{P}T_{k}\right]-C(\mathcal{E})\geq0\label{eq:agent-PC-EAIR}\\
\text{(LL) } & T\geq0,\label{eq:limited-liability}
\end{align}

The agent's incentive constraint (\ref{eq:agent-IC}) requires that
it is optimal for the agent to learn the desired experiment $\mathcal{E}_{A}$.
While the problem features both moral hazard and asymmetric information
at the same time, it turns out that the latter does not play a role.\footnote{Recall that the principal observes neither what experiment $\mathcal{E}$
the agent acquires (moral hazard) nor what realization of $\mathcal{E}$
the agent observes (asymmetric information).} This is because, as argued above, it suffices for the agent to consider
direct recommendation experiments $\mathcal{E}\in E^{K}$ and always
follow the recommendations, which is to reveals the realization of
$\mathcal{E}$ to the principal. As a result, truthtelling comes at
no additional cost, and the principal's problem is purely a moral
hazard problem. This incentive constraint will be studied extensively
in Section (\ref{sec:implementability}). 

Limited liability (\ref{eq:limited-liability}) requires that the
payment to the agent must be non-negative. This is standard in most
contracting problems as it reflects realistic institutional constraints
that the agent cannot be forced to make payments to the principal. 

The agent's participation constraint (\ref{eq:agent-PC-EAIR}) says
that the agent must do weakly better than his outside option by accepting
the contract. Under limited liability, the participation constraint
never binds because the agent can always guarantee himself a non-negative
payoff by acquiring no information.\footnote{More precisely, the agent can accept the contract, choose the uninformative
experiment $\underline{\mathcal{E}}$ with $C(\underline{\mathcal{E}})=0$,
and report a fixed $x_{k}$. Due to limited liability, his payoff
must be non-negative. The agent can potentially do better by acquiring
information and reporting optimally. This means the left hand side
of (\ref{eq:agent-PC-EAIR}) is always non-negative and PC is slack.} I include this constraint for completeness. In Section \ref{subsec:solve-cost-min},
I will consider a benchmark case without limited liability. There,
the participation constraint will bind. 

The value of the principal's problem, denoted $\kappa^{C,\mu_{0}}(\mathcal{E}_{A},\mathcal{E}_{P})$,
is the principal's indirect cost of acquiring information through
the agent. Given an $\mathcal{E}_{P}$, the principal's problem may
not be feasible for certain $\mathcal{E}_{A}$.\footnote{For example, if $\mathcal{E}_{P}$ is uninformative, the principal
is not able to incentivize the agent to acquire any informative $\mathcal{E}_{A}$.} In this case, no contract can implement the acquisition of $\mathcal{E}_{A}$
and I write $\kappa^{C,\mu_{0}}(\mathcal{E}_{A},\mathcal{E}_{P})=+\infty$
by convention. Section \ref{sec:implementability} is devoted to the
feasibility problem. Section \ref{sec:IR-cost-min} concerns how the
indirect cost changes with the principal's contractible information
$\mathcal{E}_{P}$.

\subsection{Information Cost\label{subsec:info-cost}}

The agent's information cost $C$ is assumed to belong to a general
class of smooth functions. Appendix \ref{appsec:smooth-costs} presents
the required assumptions. These assumptions are satisfied by most
cost functions studied in the literature, including mutual information
costs and log-likelihood ratio costs. 

For ease of exposition, I focus in the main text on the case where
the cost function $C$ is posterior separable, that is, $C(\mathcal{E})=\int c(\mu)\mathcal{E}(d\mu)$
for some function $c:\Delta\Omega\to\bar{\mathbb{R}}$ with $\bar{\mathbb{R}}:=\mathbb{R}\cup\{\pm\infty\}$.
$c$ is called the posterior cost function. All results remain valid
under the more general framework in Appendix \ref{appsec:smooth-costs}.
Intuitively, posterior separability says that there is a price $c(\mu)$
attached to each posterior $\mu$, and the cost of acquiring some
$\mathcal{E}$ is the probability-weighted sum of prices paid for
each posterior.

I make the following assumptions about $c$. 

\begin{assumption}
\label{assu:tech-post-sep}$c$ is lower semi-continuous and convex,
and $c(\mu_{0})=0$.
\end{assumption}
Assumption \ref{assu:tech-post-sep} is the standard technical assumption
for the agent's problem to be well-behaved. Lower semi-continuity
ensures that the agent's problem admits a maximum. Convexity implies
that $C$ is Blackwell monotone, that is, $\mathcal{E}\geq_{\text{B}}\mathcal{E}'\Rightarrow C(\mathcal{E})\geq C(\mathcal{E}')$.
In words, the more the agent learns, the higher cost he incurs. $c(\mu_{0})=0$
implies that $C(\underline{\mathcal{E}})=0$ where $\underline{\mathcal{E}}$
is an uninformative experiment, that is, learning nothing should incur
no cost.

\begin{assumption}
\label{assu:smooth-post-sep}$c$ is differentiable with derivative
$\nabla_{c}:\Delta\Omega\to\bar{\mathbb{R}}^{N}$.
\end{assumption}
Assumption \ref{assu:smooth-post-sep} defines the marginal cost of
learning $\nabla_{c}$. To see this connection, recall that $c(\mu)$
is the price for posterior $\mu$. Consider a marginal change to $\mathcal{E}$
by perturbing some $\mu\in\operatorname{Supp}\mathcal{E}$. The derivative
$\nabla_{c}(\mu)$ then describes how the price changes when the posterior
$\mu$ changes. To be specific, $\nabla_{c}(\mu)$ is an $N$-dimensional
column vector where the $n$-th entry is the marginal effect on the
price when the probability of state $\omega_{n}$ changes. As a result,
$\nabla_{c}(\mu)$ represents the marginal cost of changing the information
at posterior $\mu$.

Lastly, sometimes I assume that the marginal cost to rule out any
state is infinity. 
\begin{assumption}
\label{assu:inf-slope-post-sep} $\lim_{\mu\to\mu'}\left\Vert \nabla_{c}(\mu)\right\Vert =+\infty$
for all $\mu'\in\partial\Delta\Omega$.\footnote{I use $\partial\Delta\Omega$ to denote the boundary of $\Delta\Omega$.
This is the set of beliefs that do not have full support.}
\end{assumption}
Roughly speaking, Assumption \ref{assu:inf-slope-post-sep} says that
the posterior cost function $c$ has an infinite slope at the boundary.
It assumes away corner solutions and allows me to focus on experiments
that only induce posteriors that are fully supported. It is not required
for results in Section \ref{subsec:imp-comparison}. 

Lastly, let $\mathcal{C}_{0}$ denote the set of all cost functions
that satisfy Assumptions \ref{assu:tech-post-sep} and \ref{assu:smooth-post-sep}.
Let $\mathcal{C}_{1}$ denote the set of all cost functions that satisfy
Assumptions \ref{assu:tech-post-sep}-\ref{assu:inf-slope-post-sep}.

\section{Implementability \label{sec:implementability}}

This section studies implementability, that is, whether the principal's
problem (\ref{eq:P-problem}) has a feasible solution, and what contractible
experiment can implement a larger set of experiments for the agent
to acquire. Section \ref{subsec:imp-characterization} provides a
complete characterization of implementability. Theorem \ref{thm:imp}
gives a necessary and sufficient condition for an experiment to be
implementable. Section \ref{subsec:imp-comparison} applies this result
to compare different contractible experiments in terms of their implementable
sets. Corollary \ref{cor:full-imp} identifies a full-rank condition
under which anything feasible is implementable. Proposition \ref{prop:imp-comparison}
introduces a more general comparison result based on the column space
order: A contractible experiment has a larger implementable set than
the other if and only if its column space is larger. 

\subsection{Characterization of Implementable Sets \label{subsec:imp-characterization}}

The principal's problem (\ref{eq:P-problem}) may be infeasible because
of the agent's incentive constraint. The limited liability constraint
and the participation constraint do not pose a problem as the principal
can pay out a large constant bonus without affecting incentives.\footnote{The limited liability constraint will become relevant later in Section
\ref{sec:IR-cost-min} when I analyze the cost minimization problem.} This observation motivates the following definition of implementability. 
\begin{defn}
[Implementability] Say that a contract $T$ under $\mathcal{E}_{P}$
implements $\mathcal{E}_{A}$ if it satisfies the agent's incentive
constraint (\ref{eq:agent-IC}). Say that $\mathcal{E}_{A}$ is implementable
under $\mathcal{E}_{P}$ if there exists some contract $T$ under
$\mathcal{E}_{P}$ that implements $\mathcal{E}_{A}$.\footnote{The implementability notion defined here is weak: $\mathcal{E}_{A}$
has to be an optimal solution to the agent's problem (\ref{eq:A-problem}),
but not necessarily unique. This suffices since the agent is assumed
to breaks ties in favor of the principal. Appendix \ref{appsec:imp-strict}
defines and characterizes unique implementability which requires $\mathcal{E}_{A}$
to be the unique optimal solution. Roughly, this holds when the cost
function is strictly convex and the posteriors induced by $\mathcal{E}_{A}$
are linearly independent. This extends the unique implementability
result in \citet{matvejka2015rational}, \citet{lipnowski2022predicting}
and \citet{sharma2024procuring}. } 
\end{defn}
We are now ready to characterize which $\mathcal{E}_{A}\in E^{K}$
is implementable under a given $\mathcal{E}_{P}\in E^{M}$, fixing
some interior prior $\mu_{0}$ and cost function $C$ that satisfies
Assumptions \ref{assu:tech-post-sep}-\ref{assu:inf-slope-post-sep}.
 To this end, use $\operatorname{Col}\mathcal{E}_{P}$ to denote
the column space of $\mathcal{E}_{P}$.\footnote{The column space of a matrix $A\in\mathbb{R}^{N\times M}$ is defined
as $\operatorname{Col}A=\{Av:v\in\mathbb{R}^{M}\}$. It is the set
of all linear combinations of its columns.} 
\begin{thm}
[Characterization of Implementability]\label{thm:imp} An experiment
$\mathcal{E}_{A}\in E^{K}$ with posteriors $\left\{ x_{k}\right\} _{k=1}^{K}$
is implementable under $\mathcal{E}_{P}$ if and only if $C(\mathcal{E}_{A})<+\infty$
and 
\begin{equation}
\nabla_{k}-\nabla_{k'}\in\operatorname{Col}\mathcal{E}_{P},\forall1\leq k,k'\leq K,\label{eq:imp-cond}
\end{equation}
where $\nabla_{k}:=\nabla_{x_{k}}c$ is the marginal cost vector at
posterior $x_{k}$.
\end{thm}
The theorem says that the column space of $\mathcal{E}_{P}$ determines
what is implementable. To see why, recall that the agent solves a
standard rational inattention problem (\ref{eq:A-problem}). As a
result, his incentives to acquire information are completely determined
by his state dependent utility. $\operatorname{Col}\mathcal{E}_{P}$
is exactly the set of state-dependent utilities the principal can
generate. To see this, note that the agent's state dependent utility
from reporting some $x_{k}$ is given by $\mathcal{E}_{P}T_{k}$,
where $T_{k}\in\mathbb{R}^{M}$ specifies the payment to the agent
for each realization of $\mathcal{E}_{P}$ if the agent reports $x_{k}$.
The set of all possible $\mathcal{E}_{P}T_{k}$ is indeed the column
space of $\mathcal{E}_{P}$. 

In words, Condition (\ref{eq:imp-cond}) says that $\mathcal{E}_{A}$
is implementable if and only if the principal's contractible experiment
$\mathcal{E}_{P}$ allows her to provide the correct marginal incentives
for the agent. It is the first order condition in the agent's problem
(\ref{eq:A-problem}). To see this, consider the following perturbational
argument. Suppose the agent currently learns $\mathcal{E}_{A}$ and
considers perturbing it by moving some posterior $\mu_{k}$ to $\mu_{k}+d\mu$.
Because of Bayes plausibility, he has to also change another $\mu_{k'}$
in the opposite direction. This perturbation comes at a total marginal
cost proportional to $d\mu\cdot\left(\nabla_{k}-\nabla_{k'}\right)$.
On the other hand, since the expected payment is linear in the agent's
belief, the marginal benefit of this perturbation is proportional
to $d\mu\cdot\left(\mathcal{E}_{P}T_{k}-\mathcal{E}_{P}T_{k'}\right)$.
Optimality of $\mathcal{E}_{A}$ requires that no $d\mu$ perturbation
can increase the agent's payoff. To achieve this, the principal should
be able to equate the marginal cost to the marginal benefit by designing
the contract $T$. That is, $\nabla_{k}-\nabla_{k'}$ must be included
in $\operatorname{Col}\mathcal{E}_{P}$.  

It should be emphasized that only the difference in state dependent
utilities $\mathcal{E}_{P}T_{k}-\mathcal{E}_{P}T_{k'}$, rather than
the levels, matters for incentives. An intuitive way to see this is,
when the agent acquires more information and reports some $x_{k}$
more often, he must also report some $x_{k'}$ less often so that
the probabilities sum to one. A more fundamental reason is that the
agent's task is acquiring information. Information does not change
the marginal distribution of states. Therefore, incentives come from
how different reports lead to different payments. Any bonus that does
not depend on the agent's report, but potentially varies with state,
does not provide any incentive.

I now discuss the roles of the assumptions. Assumptions \ref{assu:tech-post-sep}
and \ref{assu:smooth-post-sep} ensure that the agent's problem is
convex and admits a solution. Under these assumptions, the first order
condition is sufficient for optimality, which proves the ``if''
direction of Theorem \ref{thm:imp}. To establish the ``only if''
direction, I also invoke Assumption \ref{assu:inf-slope-post-sep}.
This assumption ensures that the first order condition must be satisfied
as an equality.\footnote{Under Assumption \ref{assu:inf-slope-post-sep}, a $\mathcal{E}$
with a non-fully-supported posterior can never be optimal. Moving
this posterior into $\operatorname{int}\Delta\Omega$ is always an
improvement because the posterior cost function $c$ has an infinite
slope at the boundary. } Without Assumption \ref{assu:inf-slope-post-sep}, the result does
not change if every posterior induced by $\mathcal{E}_{A}$ is fully
supported.\footnote{In matrix form, this means $\mathcal{E}_{A}$ is a $N\times K$ row
matrix with no zero entries.} It matters only for $\mathcal{E}_{A}$'s with non-fully-supported
posterior. One can only perturb a non-fully-supported posterior $\mu_{k}$
in certain directions because probabilities must stay positive. This
means the first order condition may hold as an inequality. For completeness,
I present the implementability result without Assumption \ref{assu:inf-slope-post-sep}
in Appendix \ref{appsec:imp-corner-solutions}. 

I briefly sketch the proof of Theorem \ref{thm:imp}. Readers less
interested in proofs may skip to Section \ref{subsec:imp-comparison}.
All omitted proofs are presented in Appendix \ref{appsec:proof-appendix}.
The proof uses a Lagrangian approach, followed by a linear algebra
argument. Viewing $\mathcal{E}$ as a posterior distribution, the
agent solves
\begin{align}
\max_{\mathcal{E}\in E^{K}} & \:\mathbb{E}_{\mu_{k}\sim\mathcal{E}}\left[\mu_{k}\cdot\mathcal{E}_{P}T_{k}-c(\mu_{k})\right]\label{eq:A-problem-post-sep}\\
\text{s.t. } & \mathbb{E}_{\mu_{k}\sim\mathcal{E}}\left[\mu_{k}\right]=\mu_{0},\label{eq:Bayes-plausibility}
\end{align}
where (\ref{eq:Bayes-plausibility}) is the Bayes plausibility constraint.
Let $\lambda\in\mathbb{R}^{N}$ be the multipliers on the the Bayes
plausibility constraint for state $\omega_{n}$.\footnote{The interpretation of $\lambda$ is that its component $\lambda_{n}$
represents the agent's expected payoff in state $\omega_{n}$.} Let $\nabla=\begin{bmatrix}\nabla_{1} & \nabla_{2} & \cdots & \nabla_{K}\end{bmatrix}$
be the matrix of marginal cost vectors at posteriors of $\mathcal{E}_{A}$.
The first order condition requires that marginal costs equal marginal
benefits at every posterior $x_{k}$. In the matrix form, this is,
\begin{equation}
\mathcal{E}_{P}T=\nabla+\lambda\bm{1}_{1\times K}.\label{eq:imp-matrix-eqn}
\end{equation}
This is both necessary and sufficient because the problem is convex. 
\begin{lem}
\label{lem:imp-matrix-eqn} A contract $T$ under $\mathcal{E}_{P}$
implements $\mathcal{E}_{A}$ if and only if there exists $\lambda\in\mathbb{R}^{N}$
such that (\ref{eq:imp-matrix-eqn}) holds.
\end{lem}
The implementability problem is thus converted into a solvability
problem. $\mathcal{E}_{A}$ is implementable if and only if there
exists $(T,\lambda)$ that satisfies Equation (\ref{eq:imp-matrix-eqn}).
Determining the solvability of (\ref{eq:imp-matrix-eqn}) is non-trivial
since the unknowns $T$ and $\lambda$ appear on different side of
the matrix product. Luckily, \citet{baksalary1979matrix} study matrix
equations of this form and provide an equivalent condition for its
solvability. Following their result (presented in Appendix (\ref{appsec:proof-appendix})),
I show that the solvability of (\ref{eq:imp-matrix-eqn}) is equivalent
to Condition (\ref{eq:imp-cond}) in Theorem \ref{thm:imp}.

\subsection{Comparison by Implementable Sets \label{subsec:imp-comparison}}

I now apply Theorem \ref{thm:imp} and study which $\mathcal{E}_{P}$
can implement a larger set of $\mathcal{E}_{A}$'s. I will maintain
only Assumptions \ref{assu:tech-post-sep} and \ref{assu:smooth-post-sep}
throughout the results in this section. 

\subsubsection{Full Row Rank Case: Full Implementability \label{subsec:imp-full-rank}}

An important consequence of Theorem \ref{thm:imp} is that full implementability
can be achieved if $\mathcal{E}_{P}$ has full row rank.\footnote{A matrix $A\in\mathbb{R}^{N\times M}$ has full row rank if its rows
are linearly independent.} Any feasible $\mathcal{E}_{A}$ can be implemented in this case,
regardless of the prior and the cost function. This result does not
rely on Assumption \ref{assu:inf-slope-post-sep}.\footnote{Assumption \ref{assu:inf-slope-post-sep} is only used to prove the
``only if'' direction of Theorem \ref{thm:imp}. Corollary \ref{cor:full-imp}
only uses the ``if'' direction of Theorem \ref{thm:imp} and hence
does not require Assumption \ref{assu:inf-slope-post-sep}.} 
\begin{cor}
[Full Implementability]\label{cor:full-imp} If $\mathcal{E}_{P}$
has full row rank, then any feasible $\mathcal{E}_{A}$ is implementable
under it.
\end{cor}

The full row rank condition allows the principal to always provide
the correct incentives. For any state dependent utility $u_{k}\in\mathbb{R}^{N},$the
principal can always find a set of payments $T_{k}\in\mathbb{R}^{M}$
so that $u_{k}=\mathcal{E}_{P}T_{k}$ due to the full rank property.
Therefore, the principal can always equate the marginal benefit to
the marginal cost. This is true even without Assumption \ref{assu:inf-slope-post-sep},
in which case the first order condition only has to hold as an inequality.

Corollary \ref{cor:full-imp} points out that full implementability
relies only on the full row rank property of $\mathcal{E}_{P}$ rather
than its informativeness.\footnote{Even though the informativeness of $\mathcal{E}_{P}$ is irrelevant
for implementability, it will matter for cost minimization under limited
liability. I will focus on the cost minimization problem in Section
\ref{sec:IR-cost-min}. Roughly speaking, when $\mathcal{E}_{P}$
becomes less informative, the payments must become more dispersed
to create the same utility, and the principal has to pay larger bonuses
overall. } This extends and steamlines the full implementability result in \citet{whitmeyer2022buying}
and \citet{sharma2024procuring}. They study a similar problem but
with contractible states and show that full implementability can be
achieved.\footnote{In the language of my model, they study a problem where $\mathcal{E}_{P}$
is fixed to be the experiment that perfectly reveals the state.} In contrast, Corollary \ref{cor:full-imp} says, full implementability
can be achieved even though $\mathcal{E}_{P}$ is barely informative
about the states, provided that it has full row rank.\footnote{One way to see the connection to \citet{whitmeyer2022buying} and
\citet{sharma2024procuring} is as follows. When $\mathcal{E}_{P}$
has full row rank, there exist payment rules that lead to state-dependent
utility profiles $\bm{e}_{n}$, which provides an expected utility
of one in state $\omega_{n}$ and zero in other states. The principal
can then use $\bm{e}_{n}$'s as the building blocks to construct the
contract as if the states are contractible.} 

The full row rank condition in Corollary \ref{cor:full-imp} is closely
related to other full rank conditions in the literature. It is an
identifiability condition similar to \citet{fudenberg1994folk}. They
study repeated games with imperfect public monitoring. Their individual
full rank condition ensures that every unilateral deviation can be
identified from the distribution of the public signal.\footnote{More precisely, the individual full rank condition requires, for every
player, holding the other players strategies fixed, the distribution
of public signals resulting from different action by that player must
be linearly independent.} My full row rank condition does the same thing, except the agent's
action is to choose a distribution of posteriors: The principal can
identify any belief the agent may have using the realization of $\mathcal{E}_{P}$.
To see this, the full row rank property says that the distributions
of realizations sent in every state are linearly independent. As a
result, no two posterior distributions over the states can result
in the same distribution of realizations. As a result, the principal
can deter any deviation from the desired experiment $\mathcal{E}_{A}$
by carefully designing the contract. 

It is also related to, though different from, the full rank condition
in \citet{cremer1988full}. They study an auction design problem and
show that full surplus extraction can be achieved if every bidder
has linearly independent interim beliefs across their types.\footnote{In \citeauthor{cremer1988full}, full surplus extraction is possible
because the agent's participation constraint is imposed at the interim
stage. It is not possible in my model because of limited liability
which is a form of ex post participation constraint. This will become
clear in Section \ref{sec:IR-cost-min}. } Their full rank condition is similar to mine in the sense that both
allow the principal to construct any state dependent utility. Yet,
neither of the two conditions implies the other. In the language of
my model, their full rank condition requires the agent's interim beliefs
about the realization of the principal's experiment to be linearly
independent across agent types.\footnote{Formally, they require $\mathcal{E}_{P}^{\intercal}\mathcal{E}_{A}$
to have full column rank. The symbol $\intercal$ stands for transpose.
To see this, given $\mathcal{E}_{A}$, the agent's posterior beliefs
about the states are $x_{k}$'s which are proportional to the columns
of $\mathcal{E}_{A}$. The agent's beliefs about the realization of
$\mathcal{E}_{P}$ are therefore $\mathcal{E}_{P}^{\intercal}x_{k}$.
\citet{cremer1988full}'s full rank condition requires linear independence
of $\mathcal{E}_{P}^{\intercal}x_{k}$'s, which corresponds to $\mathcal{E}_{P}^{\intercal}\mathcal{E}_{A}$
having full column rank.} My full row rank condition is only about the principal's information
$\mathcal{E}_{P}$. It places no restriction on $\mathcal{E}_{A}$.

Next, I provide several sufficient conditions for full implementability.
The next corollary says that if the state space is binary, any informative
$\mathcal{E}_{P}$ has full row rank and achieves full implementability.
The proof is immediate. 
\begin{cor}
[Binary State]\label{cor:binary-state} If $N=2$ and $\mathcal{E}_{P}$
is informative, then any feasible $\mathcal{E}_{A}$ is implementable.
\end{cor}
The next result concerns uniform random noisy. Define the following
class of experiments. 
\begin{defn}
\label{def:random-noise} Say that $\mathcal{E}_{P}$ has uniform
random noise if, for every state $\omega_{n}$, there is a unique
and distinct realization $y_{m}$ with the largest conditional probability,
and all other realizations have the same conditional probability.
\end{defn}
The definition says that $\mathcal{E}_{P}$ is constructed by adding
state-dependent uniform random noise to a fully revealing experiment.
Such experiments must have at least $N$ realizations. The next corollary
says that if $\mathcal{E}_{P}$ has uniform random noise, then it
has full row rank and achieves full implementability. 
\begin{cor}
[Random Noise]\label{cor:random-noise} If $\mathcal{E}_{P}$ has
uniform random noise, then any feasible $\mathcal{E}_{A}$ is implementable.
\end{cor}
It should be noted that the uniformness of the noise is important.
The following alone does not imply full row rank: For every state
$\omega_{n}$, there is a unique and distinct realization $y_{m}$
with the largest conditional probability. Below is a counterexample. 
\begin{example}
Both experiments $\mathcal{E}_{1}$ and $\mathcal{E}_{2}$ have a
unique and distinct realization $y_{m}$ with the largest conditional
probability for each state. Experiment $\mathcal{E}_{1}$ has uniform
random noise but $\mathcal{E}_{2}$ does not. $\mathcal{E}_{1}$ has
full row rank, but $\mathcal{E}_{2}$ does not and one can verify
this by computing the determinant of $\mathcal{E}_{2}$.\footnote{A square matrix has full row rank if and only if it is invertible,
which in turn is equivalent to having a non-zero determinant. The
determinant of $\mathcal{E}_{2}$ is zero.}
\[
\mathcal{E}_{1}=\begin{array}{c}
\frac{1}{6}\\
\frac{1}{7}\\
\frac{1}{8}
\end{array}\begin{bmatrix}3 & 1 & 1 & 1\\
1 & 2 & 1 & 1\\
1 & 1 & 4 & 1
\end{bmatrix};\mathcal{E}_{2}=\begin{array}{c}
\frac{1}{2+\sqrt{2}}\\
\frac{1}{2+2\sqrt{2}}\\
\frac{1}{2+\sqrt{2}}
\end{array}\begin{bmatrix}2 & \sqrt{2} & 0\\
\sqrt{2} & 2 & \sqrt{2}\\
0 & \sqrt{2} & 2
\end{bmatrix}.
\]
\end{example}

\subsubsection{General Case: Column Space Order}

More generally, not all experiments are implementable if $\mathcal{E}_{P}$
does not have full row rank. From Theorem \ref{thm:imp}, the implementability
of certain $\mathcal{E}_{A}$ depends on both the column space of
$\mathcal{E}_{P}$ and the exact shape of the cost function at $\mathcal{E}_{A}$.
There is no simple characterization for the implementable set of experiments.
Appendix \ref{appsec:imp-deficient-rank} provides an example.

On the other hand, I can still compare different $\mathcal{E}_{P}$'s
in terms of which one can always implement a larger set of $\mathcal{E}_{A}$'s.
Formally, let $\mathcal{I}^{\mu_{0},C}(\mathcal{E}_{P})$ denote the
set of implementable experiments when the prior is $\mu_{0}$, the
agent's cost function is $C$, and the principal has $\mathcal{E}_{P}$
as the contractible experiment. Consider any two experiments $\mathcal{E}_{P}$
and $\mathcal{E}_{P}'$. The goal is to identify the condition under
which $\mathcal{I}^{C,\mu_{0}}(\mathcal{E}_{P})\supseteq\mathcal{I}^{C,\mu_{0}}(\mathcal{E}_{P}')$
holds for any cost function $C\in\mathcal{C}_{0}$ and interior prior
$\mu_{0}$.\footnote{$\mathcal{C}_{0}$ is the set of all cost functions that satisfy Assumptions
\ref{assu:tech-post-sep} and \ref{assu:smooth-post-sep}. For the
comparison of implementable set, I do not need Assumption \ref{assu:inf-slope-post-sep}. } 

The column space of $\mathcal{E}_{P}$ plays an crucial role in determining
the implementable set, as is clear from Theorem \ref{thm:imp} and
Corollary \ref{cor:full-imp}. This motivates following definition
of the column space order.
\begin{defn}
[Column Space Order] Say that experiment $\mathcal{E}_{P}\in E^{M}$
dominates $\mathcal{E}_{P}'\in E^{M'}$ in the column space order,
denoted $\mathcal{E}_{P}\geq_{\text{Col}}\mathcal{E}_{P}'$, if $\operatorname{Col}\mathcal{E}_{P}\supseteq\operatorname{Col}\mathcal{E}_{P}'$.
Say that $\mathcal{E}_{P}\in E^{M}$ dominates $\mathcal{E}_{P}'\in E^{M'}$
in the strict column space order, denoted $\mathcal{E}_{P}>_{\text{Col}}\mathcal{E}_{P}'$,
if $\operatorname{Col}\mathcal{E}_{P}\supsetneq\operatorname{Col}\mathcal{E}_{P}'$.
\end{defn}
The column space order is first introduced by \citet{azrieli2022elicitability}.
It can be viewed as a relaxation of the Blackwell order. It is easy
to see that $\mathcal{E}_{P}\geq_{\text{Col}}\mathcal{E}_{P}'$ if
and only if there exists some matrix $G\in\mathbb{R}^{M\times M'}$
such that $\mathcal{E}_{P}'=\mathcal{E}_{P}G$, as opposed to a garbling
matrix required by Blackwell. In some sense, the column space order
only captures the qualitative information contained in $\mathcal{E}_{P}$. 

The column space order characterizes which experiment always has a
larger implementable set. The following proposition formalizes this. 
\begin{prop}
[Comparison of Implementable Set]\label{prop:imp-comparison} $\mathcal{E}_{P}\geq_{\text{Col}}\mathcal{E}_{P}'$
if and only if $\mathcal{I}^{C,\mu_{0}}(\mathcal{E}_{P})\supseteq\mathcal{I}^{C,\mu_{0}}(\mathcal{E}_{P}')$
for any $C\in\mathcal{C}_{0}$ and $\mu_{0}\in\operatorname{int}\Delta\Omega$.
$\mathcal{E}_{P}>_{\text{Col}}\mathcal{E}_{P}'$ if and only if the
above holds and $\mathcal{I}^{C,\mu_{0}}(\mathcal{E}_{P})\supsetneq\mathcal{I}^{C,\mu_{0}}(\mathcal{E}_{P}')$
for some $C\in\mathcal{C}_{0}$ and $\mu_{0}\in\operatorname{int}\Delta\Omega$.
\end{prop}
Intuitively, a larger column space means the principal can generate
a larger set of state dependent utilities, hence enlarging the implementable
set. The converse follows by construction: If the column space inclusion
does not hold, there exists a state dependent utility that is feasible
only under the second experiment $\mathcal{E}_{P}'$. Using this utility,
one can construct a cost function such that certain learning is implementable
only under $\mathcal{E}_{P}'$ but not $\mathcal{E}_{P}$. 

While the column space answers the question of implementability, it
does not speak to the costs of implementation. The next section examines
cost minimization and studies which experiments provide incentives
more efficiently.

\section{Cost Minimization \label{sec:IR-cost-min}}

Now that we know what can be implemented, I turn to the principal's
cost minimization problem. The value of this problem is the principal's
indirect cost of acquiring information through the agent. This section
studies the question of what contractible information always leads
to a lower implementation cost. As the first step, in Section \ref{subsec:imp-contracts},
Theorem \ref{thm:imp-contract-decomposition} characterizes the set
of all contracts that implements a given experiment. This allows me
to reduce the principal's problem as a linear program, and obtain
a solution in Section \ref{subsec:solve-cost-min}. Section \ref{subsec:LL-cost-comparison}
defines and studies a new order on experiments based on the implementation
costs. Prop \ref{prop:LL-cost-comparison} provides a general sufficient
condition for this order, and Prop \ref{prop:binary-LL-characterization}
provides a complete characterization in the binary-binary case. I
maintain Assumptions \ref{assu:tech-post-sep}-\ref{assu:inf-slope-post-sep}
thoughout the section. 

\subsection{Contracts Implementing a Given $\mathcal{E}_{A}$ \label{subsec:imp-contracts}}

I now characterize the set of all contracts under $\mathcal{E}_{P}\in E^{M}$
that implement a given $\mathcal{E}_{A}\in E^{K}$, fixing some interior
prior $\mu_{0}$ and cost function $C$ that satisfies Assumptions
\ref{assu:tech-post-sep}-\ref{assu:inf-slope-post-sep}. The following
theorem characterizes all such contracts $T\in\mathbb{R}^{M\times K}$.\footnote{Recall that a contract $T\in\mathbb{R}^{M\times K}$ specifies the
payment to the agent given the realizations $\{y_{m}\}_{m=1}^{M}$
of $\mathcal{E}_{P}$ and the reports of the agent $\{x_{k}\}_{k=1}^{K}$.
Every row of $T$ represents a realization of $\mathcal{E}_{P}$ and
every column represents the agent's reports. $T_{k}$ denotes the
$k$-th column of $T$. It specifies the payments following each realization
of $\mathcal{E}_{P}$ when the agent reports $x_{k}$.}
\begin{thm}
[Shape of Contracts]\label{thm:imp-contract-decomposition} Suppose
$\mathcal{E}_{A}$ is implementable under $\mathcal{E}_{P}$. Any
contract $T$ under $\mathcal{E}_{P}$ that implements $\mathcal{E}_{A}$
takes the following form, 
\begin{equation}
T_{k}=\mathcal{E}_{P}^{-}\nabla_{k}+Z+W_{k},\forall1\leq k\leq K\label{eq:imp-payment-rule}
\end{equation}
where $\mathcal{E}_{P}^{-}$ is the pseudo-inverse of $\mathcal{E}_{P}$,
$Z\in\mathbb{R}^{M}$ and $W_{k}\in\mathbb{R}^{M}$ can take any values
as long as $\mathcal{E}_{P}W_{k}=\boldsymbol{0},\forall1\leq k\leq K$.
\end{thm}
Theorem \ref{thm:imp-contract-decomposition} says that any contract
that implements $\mathcal{E}_{A}$ can be decomposed into three parts.
The first part $\mathcal{E}_{P}^{-}\nabla_{k}$ provides the correct
incentives for the agent to learn $\mathcal{E}_{A}$. $\nabla_{k}$
is the marginal cost vector of changing beliefs at posterior $x_{k}$
of $\mathcal{E}_{A}$. If the principal can contract on the state,
she would like to set the agent's state dependent utility equal to
the marginal cost $\nabla_{k}$. However, since she only has access
to $\mathcal{E}_{P}$, she can only generate state dependent utilities
in $\operatorname{Col}\mathcal{E}_{P}$. The best she can do is to
approximate $\nabla_{k}$ by its orthogonal projection onto $\operatorname{Col}\mathcal{E}_{P}$
using payments $\mathcal{E}_{P}^{-}\nabla_{k}$.\footnote{This is due to the following property of the psuedo-inverse: For any
vector $v$, the product $\mathcal{E}_{P}^{-}v$ gives the minimum-norm
solution to the least-squares problem $\mathcal{E}_{P}x=v$. As a
result, $\mathcal{E}_{P}\mathcal{E}_{P}^{-}v$ is the orthogonal projection
of $v$ onto the column space of $\mathcal{E}_{P}$.}  

 The remaining two parts are payments that do not affect incentives,
but they determine the agent's rents. $Z$ is a bonus (or punishment,
if negative) term that depends only on the realization of $\mathcal{E}_{P}$.
$Z\in\mathbb{R}^{M}$ specifies the payment for each realization of
$\mathcal{E}_{P}$ is. Such bonuses obviously do not affect incentives
as they do not depend on the agent's report. 

$W_{k}$'s are a collection of side bets with an expected payment
of zero conditional on every state. Every $W_{k}$ represents a side
bet where the principal pays the agent $W_{k}^{m}$ if the agent reports
$x_{k}$ and $\mathcal{E}_{P}$'s realization is $y_{m}$. These side
bets do not affect incentives though they depend on the agent's report.
This is because $\mathcal{E}_{P}W_{k}=0$, which means the expected
value from $W_{k}$ is zero in every state. As a result, information
about the states never changes the value of such side bets, and they
do not matter for incentives. These side bets are non-trivial only
when $\mathcal{E}_{P}$ has deficient row rank. If $\mathcal{E}_{P}$
has full row rank, we necessarily have $W_{k}=\boldsymbol{0}$. 

\subsection{Characterization of Implementation Costs \label{subsec:solve-cost-min}}

I now characterize the implementation cost. This consists of two parts,
namely, the first best cost which is the agent's information cost
to acquire the information, and the agency rent, which is the additional
payment to the agent because of moral hazard. Agency rents exist here
because of limited liability. Without it, the principal can always
leave no rents to the agent and any implementable $\mathcal{E}_{A}$
can be implemented at the first best cost $C(\mathcal{E}_{A})$.\footnote{This is because principal can start from any $T$ that implements
$\mathcal{E}_{A}$ and make the participation constraint (\ref{eq:agent-PC-EAIR})
bind by reducing the payment uniformly. Details are provided in Appendix
(\ref{appsec:proof-appendix}).}

Under limited liability, the principal has to optimally choose where
the limited liability constraint should bind to minimize the implemetation
cost. When $\mathcal{E}_{P}$ has full row rank, the optimal contract
and the cost can be fully characterized.\footnote{When $\mathcal{E}_{P}$ has deficient row rank, the principal has
to solve a more complicated linear program to obtain the optimal contract.} To compactly state this result, I introduce the $\operatorname{rowmin}$
notation. For any matrix $A\in\mathbb{R}^{N\times K}$, let $\operatorname{rowmin}\left(A\right)\in\mathbb{R}^{N}$
denote the vector of row-wise minima of $A$, that is, $\left[\operatorname{rowmin}\left(A\right)\right]_{n}:=\min_{k}A_{nk}$. 
\begin{prop}
[Optimal Contract]\label{prop:LL-opt-contract} Suppose $\mathcal{E}_{P}$
has full row rank. The optimal contract $T$ is uniquely given by
$T_{k}=\mathcal{E}_{P}^{-}\nabla_{k}+Z$ with $Z:=-\operatorname{rowmin}\left(\mathcal{E}_{P}^{-}\nabla\right)$.
The minimum cost is $\kappa^{C,\mu_{0}}(\mathcal{E}_{A},\mathcal{E}_{P})=C(\mathcal{E}_{A})-\mu_{0}\cdot\mathcal{E}_{P}\operatorname{rowmin}\left(\mathcal{E}_{P}^{-}\nabla\right)$.
\end{prop}
The idea is simple. We know from Theorem \ref{thm:imp-contract-decomposition}
that every contract $T$ that implements $\mathcal{E}_{A}$ must take
the form of (\ref{eq:imp-payment-rule}). When $\mathcal{E}_{P}$
has full row rank, the side bets $W_{k}=\boldsymbol{0}$ must be trivial.
The principal optimally picks the bonuses $Z$ to be the minus of
the row-wise minima of $\mathcal{E}_{P}^{-}\nabla$ so that, for each
realization of $\mathcal{E}_{P}$, at least one report $x_{k}$ gets
a zero payment. 

The principal's indirect cost $\kappa^{C,\mu_{0}}(\mathcal{E}_{A},\mathcal{E}_{P})$
derived in Proposition \ref{prop:LL-opt-contract} generalizes that
of \citet{sharma2024procuring}, who study the same problem with the
restriction that $\mathcal{E}_{P}$ fully reveals the state. Several
key features of the indirect cost function persist when I allow more
general $\mathcal{E}_{P}$. In particular, the indirect cost function
is lower semi-continuous, but it is neither posterior separable nor
Blackwell monotone.

\subsection{Comparison by Implementation Costs \label{subsec:LL-cost-comparison}}

We are now ready to compare experiments in terms of the implementation
costs. I define the following order as the main object of interest
in this section. 
\begin{defn}
[Indirect Cost Order] Say that $\mathcal{E}_{P}$ dominates $\mathcal{E}_{P}'$
in the indirect cost order, denoted $\mathcal{E}_{P}\geq_{\text{K}}\mathcal{E}_{P}'$,
if $\kappa^{C,\mu_{0}}(\mathcal{E}_{A},\mathcal{E}_{P})\leq\kappa^{C,\mu_{0}}(\mathcal{E}_{A},\mathcal{E}_{P}')$
for all $C\in\mathcal{C}_{1}$, $\mu_{0}\in\operatorname{int}\Delta\Omega$,
and $\mathcal{E}_{A}\in E$.\footnote{When $\mathcal{E}_{A}$ is not implementable, the indirect cost is
set to $+\infty$. If $\mathcal{E}_{A}$ is not implementable under
both $\mathcal{E}_{P}$ and $\mathcal{E}_{P}'$, I adopt the convention
that the inequality $\kappa^{C,\mu_{0}}(\mathcal{E}_{A},\mathcal{E}_{P})\leq\kappa^{C,\mu_{0}}(\mathcal{E}_{A},\mathcal{E}_{P}')$
still holds. } 
\end{defn}

\subsubsection{Sufficient Condition: The Conic Span Order}

I first provide sufficient conditions for $\geq_{\text{K}}$. While
it is obvious that the Blackwell order implies the indirect cost order,
it is not a necessary requirement.\footnote{If $\mathcal{E}_{P}\geq_{\text{B}}\mathcal{E}_{P}'$, the principal
can garble $\mathcal{E}_{P}$ herself and construct any contract under
$\mathcal{E}_{P}'$ as a contract under $\mathcal{E}_{P}$. As a result,
$\mathcal{E}_{P}\geq_{\text{B}}\mathcal{E}_{P}'$ implies $\mathcal{E}_{P}\geq_{\text{K}}\mathcal{E}_{P}'$. } In fact, there exists weaker sufficient conditions than Blackwell.
To state the sufficient condition, I introduce the definition of the
conic span of a matrix. Formally, for any matrix $A\in\mathbb{R}^{N\times M}$,
the conic span $\operatorname{Cone}(A):=\{Av:v\in\mathbb{R}_{+}^{M}\}$
is the set of all non-negative linear combinations of its columns.
The conic span order, first appeared in \citet{xia2025comparisons},
is defined as follows. 
\begin{defn}
[Conic Span Order] Say that $\mathcal{E}_{P}$ dominates $\mathcal{E}_{P}'$
in the conic span order, denoted $\mathcal{E}_{P}\geq_{\text{Cone}}\mathcal{E}_{P}'$,
if $\operatorname{Cone}(\mathcal{E}_{P})\supseteq\operatorname{Cone}(\mathcal{E}_{P}')$.
Say that $\mathcal{E}_{P}$ dominates $\mathcal{E}_{P}'$ in the strict
conic span order, denoted $\mathcal{E}_{P}>_{\text{Cone}}\mathcal{E}_{P}'$,
if $\operatorname{Cone}(\mathcal{E}_{P})\supsetneq\operatorname{Cone}(\mathcal{E}_{P}')$. 
\end{defn}
The conic span order can also be viewed as a relaxation of the Blackwell
order. Specifically, $\mathcal{E}_{P}\geq_{\text{Cone}}\mathcal{E}_{P}'$
is equivalent to $\mathcal{E}_{P}'=\mathcal{E}_{P}G$ for some $G\geq\boldsymbol{0}$.\footnote{The inequality $G\geq\boldsymbol{0}$ is entrywise, that is, it requires
every entry of the matrix $G$ to be weakly positive. } It should be noted that when $\mathcal{E}_{P}$ does not have full
rank, there may be more than one $G$ that satisfies $\mathcal{E}_{P}'=\mathcal{E}_{P}G$.
The definition only requires that there exists some $G\geq\boldsymbol{0}$
such that $\mathcal{E}_{P}'=\mathcal{E}_{P}G$. 

The following result says that the conic span order is sufficient
for indirect cost dominance. It applies to all experiments, regardless
of whether they have full rank. 
\begin{prop}
[Sufficient Condition for Indirect Cost Dominance]\label{prop:LL-cost-comparison}
$\mathcal{E}_{P}\geq_{\text{Cone}}\mathcal{E}_{P}'$ implies $\mathcal{E}_{P}\geq_{\text{K}}\mathcal{E}_{P}'$.
 
\end{prop}
The intuition of Proposition \ref{prop:LL-opt-contract} can be understood
in terms of the possible state dependent utility profiles the principal
can generate using $\mathcal{E}_{P}$ and $\mathcal{E}_{P}'$. The
conic span of an experiment $\mathcal{E}_{P}$ represents the set
of all possible state-dependent utility profiles the principal can
generate using non-negative payment rules. The requirement of non-negativity
is related to the agent's limited liability constraint. Using this
interpretation, Proposition \ref{prop:LL-opt-contract} should become
straightforward. If $\mathcal{E}_{P}$ dominates $\mathcal{E}_{P}'$
in the conic span order, the principal has access to a larger set
of state dependent utility profiles under $\mathcal{E}_{P}$. Therefore,
the implementation cost of any $\mathcal{E}_{A}$ should be lower
under $\mathcal{E}_{P}$. 

One can also understand Proposition \ref{prop:LL-opt-contract} using
the payment rules directly. Roughly speaking, when $\mathcal{E}_{P}\geq_{\text{Cone}}\mathcal{E}_{P}'$,
the principal can use a less dispersed payment rule under $\mathcal{E}_{P}$,
resulting a lower cost. To see this, consider the simplest case where
both $\mathcal{E}_{P}$ and $\mathcal{E}_{P}'$ has full row rank.
We can write $\mathcal{E}_{P}'=\mathcal{E}_{P}G$ for some $G$.\footnote{Such a $G$ exists because $\mathcal{E}_{P}$ has full row rank.}
Suppose the principal wants to implement some $\mathcal{E}_{A}$ with
a marginal cost matrix $\nabla$. Applying Proposition \ref{prop:LL-opt-contract},
the difference in the implementation costs is given by 
\begin{align}
\kappa^{C,\mu_{0}}(\mathcal{E}_{A},\mathcal{E}_{P}')-\kappa^{C,\mu_{0}}(\mathcal{E}_{A},\mathcal{E}_{P}) & =\mu_{0}\cdot\mathcal{E}_{P}\left[\operatorname{rowmin}\left(GV\right)-G\operatorname{rowmin}\left(V\right)\right]\geq0,\label{eq:EPIR-cost-inequality}
\end{align}
where $V=G^{-}\mathcal{E}_{P}^{-}\nabla$.\footnote{$GV=GG^{-}\mathcal{E}_{P}^{-}\nabla=\mathcal{E}_{P}^{-}\nabla$ because
$GG^{-}=I$ is the identity matrix. This is because $G$ must also
have full row rank due to the full row rank assumptions on $\mathcal{E}_{P}$
and $\mathcal{E}_{P}'$.} This inequality holds because the term in the square bracket is always
positive. because the minimum of convex combinations is less extreme
than convex combinations of the minimum.\footnote{Since $G\geq0$, we can normalize the above inequality by dividing
every row by the corresponding row sum of $G$. }
\begin{equation}
\min GV-G\min V\geq\bm{0},\label{eq:EPIR-interim-cost-inequality}
\end{equation}

While conic span dominance implies indirect cost dominance, the converse
is not true. Example \ref{exa:EPIR-not-Blackwell} is a counterexample
where $\mathcal{E}_{1}$ and $\mathcal{E}_{2}$ are not ranked by
the conic span or the Blackwell orders, but $\mathcal{E}_{1}\geq_{\text{K}}\mathcal{E}_{2}$.

\subsubsection{Characterization in the Binary-Binary Case}

I can fully characterize the indirect cost order $\geq_{\text{K}}$
in the binary-binary case. From now on, I assume the state space is
binary, and use $\mathcal{E}_{P},\mathcal{E}_{P}'\in E^{2}$ to denote
experiments with binary realizations. 

To state the result, I first define the following likelihood ratios.
Take any $\mathcal{E}_{P},\mathcal{E}_{P}'\in E^{2}$. $\mathcal{E}_{P}$
has two realizations $y_{1}$ and $y_{2}$, each corresponding to
a likelihood ratio given by 
\[
\ell_{1}:=\dfrac{\mathcal{E}_{P}(y_{1}\mid\omega_{2})}{\mathcal{E}_{P}(y_{1}\mid\omega_{1})},\;\ell_{2}:=\dfrac{\mathcal{E}_{P}(y_{2}\mid\omega_{2})}{\mathcal{E}_{P}(y_{2}\mid\omega_{1})}.
\]
The reciprocals $1/\ell_{1}$ and $1/\ell_{2}$ are also likelihood
ratios but with the conditional probabilities of $\omega_{1}$ in
the numerator. Similarly, let $\ell_{1}'$ and $\ell_{2}'$ be the
corresponding likelihood ratios for $\mathcal{E}_{P}'$. Without loss,
assume $\ell_{1}\leq1\leq\ell_{2}$ and $\ell_{1}'\leq1\leq\ell_{2}'$.\footnote{The two likelihood ratios cannot be both less than or greater than
one because the probabilities must add up.} The following proposition characterizes $\geq_{\text{K}}$ in the
binary-binary case using the differences in the likelihood ratios. 
\begin{prop}
[Binary-Binary Indirect Cost]\label{prop:binary-LL-characterization}
Suppose $N=2$, and $\mathcal{E}_{P},\mathcal{E}_{P}'\in E^{2}$ are
binary experiments. $\mathcal{E}_{P}\geq_{\text{K}}\mathcal{E}_{P}'$
if and only if 
\begin{align*}
\ell_{2}-\ell_{1} & \geq\ell_{2}'-\ell_{1}',\\
\dfrac{1}{\ell_{1}}-\dfrac{1}{\ell_{2}} & \geq\dfrac{1}{\ell_{1}'}-\dfrac{1}{\ell_{2}'}.
\end{align*}
\end{prop}
\begin{figure}[th]
\begin{centering}
\caption{Likelihood Ratio Differences \label{fig:LR-diff}}
\par\end{centering}
\begin{centering}
\medskip{}
\par\end{centering}
\begin{centering}
\begin{minipage}[t]{0.4\columnwidth}%
\begin{center}
\includegraphics[width=1\columnwidth]{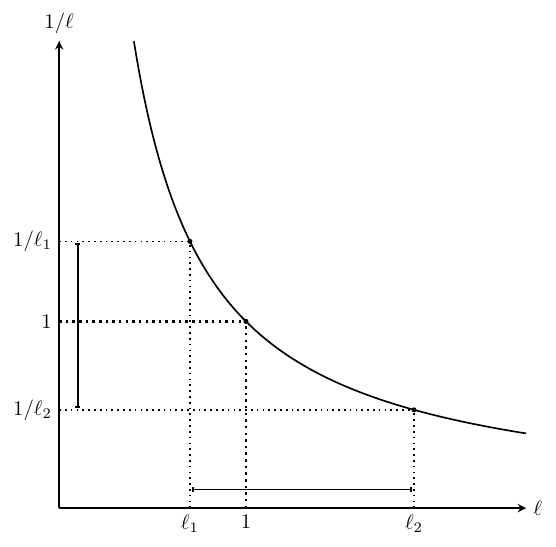}
\par\end{center}%
\end{minipage}
\par\end{centering}
\centering{}%
\begin{minipage}[t]{0.9\columnwidth}%
\begin{singlespace}
{\scriptsize Notes: This figure depicts what matters to compare experiments
in the binary-binary case according to Proposition \ref{prop:binary-LL-characterization}.
The horizontal axis represents the likelihood ratio $\ell$ of a realization,
and the vertical axis represents the reciprocal likelihood ratio $1/\ell$. }{\scriptsize\par}
\end{singlespace}

\end{minipage}
\end{figure}

Proposition \ref{prop:binary-LL-characterization} says that an experiment
dominates another if and only if the differences in the likelihood
ratios is larger under the first experiment. Figure \ref{fig:LR-diff}
illustrates the condition in Proposition \ref{prop:binary-LL-characterization}.
Dominance requires the difference in the likelihood ratios and their
reciprocals, which are also likelihood ratios, to both be larger at
the same time.

To understand the intuition, the differences in likelihood ratios
$\ell_{2}-\ell_{1}$ and $1/\ell_{1}-1/\ell_{2}$ are the marginal
products of rents. They specify for every dollar of rents paid to
the agent, how much incentives are provided. This comes from a geometric
argument that extends the intuition in Example \ref{exa:EPIR-not-Blackwell}.
In the binary-binary case, it is without loss to focus on implementing
$\mathcal{E}_{A}$'s with two realizations. In this case, the principal
should choose a pay-if-correct contract $T\in\mathbb{R}_{+}^{2\times2}$
with
\[
T=\begin{array}{c}
y_{1}\\
y_{2}
\end{array}\overset{\begin{array}{cc}
x_{1} & x_{2}\end{array}}{\begin{bmatrix}t_{1} & 0\\
0 & t_{2}
\end{bmatrix}}.
\]
That is, the principal pays the agent only if the agent's report agrees
with the realization of $\mathcal{E}_{P}$. What matters for the agent's
learning incentives is his expected payoff as a function of his posterior
belief about the state. As shown in Figure \ref{fig:value-func},
the agent's incentive is determined by $\Delta u_{n}=u(\omega_{n},x_{1})-u(\omega_{n},x_{2})$
for $n=1,2$. On the other hand, $r_{1}$ and $r_{2}$ are the agent's
rent when the state is $\omega_{1}$ and $\omega_{2}$. Call a pair
$(\Delta u_{1},\Delta u_{2})$ an incentive profile as it determines
the agent's learning incentives.

To show that $\mathcal{E}_{1}$ provides cheaper incentives than $\mathcal{E}_{2}$
in all circumstances, we need to verify that, for any $(\Delta u_{1},\Delta u_{2})$,
the rents paid to the agent are lower under $\mathcal{E}_{1}$ than
$\mathcal{E}_{2}$. This does not involve the agent’s information
cost function because the principal's information can only affect
the cost to provide incentives. In the binary-binary setting, incentives
are linear in the contracts. This is because we always use pay-if-correct
contracts and the limited liability constraints always bind when the
agent's report and the principal's experiment realization mismatch.\footnote{This is not true in general because the limited liability constraints
may bind at different places if we do not have the binary-binary structure.} This means we can analyze the simpler incentive profiles $(\Delta u_{1},0)$
and $(0,\Delta u_{2})$ separately. 

I now illustrate how rents are related to incentives. Consider the
contract to provide incentives $(\Delta u_{1},0)$. As illustrated
in Figure \ref{fig:LR-characterization} Panel A, we can see that
the rents and the incentive $\Delta u_{1}$ are connected via the
likelihood ratios by 
\[
\dfrac{r_{2}}{r_{1}}=\ell_{2},\;\dfrac{r_{2}}{\Delta u_{1}+r_{1}}=\ell_{1},
\]
which means to generate a unit of $\Delta u_{1}$, $\frac{\ell_{1}}{\ell_{2}-\ell_{1}}$
units of $r_{1}$ and $\frac{\ell_{2}\ell_{1}}{\ell_{2}-\ell_{1}}$
units of $r_{2}$ are needed. In other words, for every unit of rent
$r_{2}$ spent, $1/\ell_{1}-1/\ell_{2}$ units of $\Delta u_{1}$
can be provided.\footnote{To generate $\Delta u_{1}$, for every unit of $r_{2}$, we have to
spend $\ell_{1}$ units of $r_{1}$, as illustrated in Figure \ref{fig:LR-characterization}
Panel B. Therefore, it suffices to keep track of the marginal product
of $r_{2}$ when producing $\Delta u_{1}$.} 

On the other hand, to provide incentives $(0,\Delta u_{2})$, from
Figure \ref{fig:LR-characterization} Panel B, we have 
\[
\dfrac{r_{2}}{r_{1}}=\ell_{1},\;\dfrac{\Delta u_{2}+r_{2}}{r_{1}}=\ell_{2}.
\]
which means, for every unit of $\Delta u_{2}$, we need $\frac{1}{\ell_{2}-\ell_{1}}$
units of $r_{1}$ and $\frac{\ell_{1}}{\ell_{2}-\ell_{1}}$ units
of $r_{2}$. In other words, for every unit of $r_{1}$ spent, $\ell_{2}-\ell_{1}$
units of $\Delta u_{2}$ can be provided.\footnote{To generate $\Delta u_{2}$, for every unit of $r_{1}$, we have to
spend $\ell_{2}$ units of $r_{2}$, as illustrated in Figure \ref{fig:LR-characterization}
Panel A. Therefore, it suffices to keep track of the marginal product
of $r_{1}$ when producing $\Delta u_{2}$.}

We can now compare experiments $\mathcal{E}_{P}$ and $\mathcal{E}_{P}'$.
Since the principal may provide any incentive profile $(\Delta u_{1},\Delta u_{2})$,
for $\mathcal{E}_{P}$ to dominate $\mathcal{E}_{P}'$, I need the
marginal product of rents to be higher, which is the condition in
Proposition \ref{prop:binary-LL-characterization}.

\begin{figure}[th]
\begin{centering}
\caption{Marginal Product of Rents \label{fig:LR-characterization}}
\par\end{centering}
\begin{centering}
\medskip{}
\par\end{centering}
\begin{centering}
\begin{minipage}[t]{0.48\columnwidth}%
\begin{center}
{\small Panel A: Provide $(\Delta u_{1},0)$}{\small\par}
\par\end{center}
\begin{center}
\includegraphics[width=1\columnwidth]{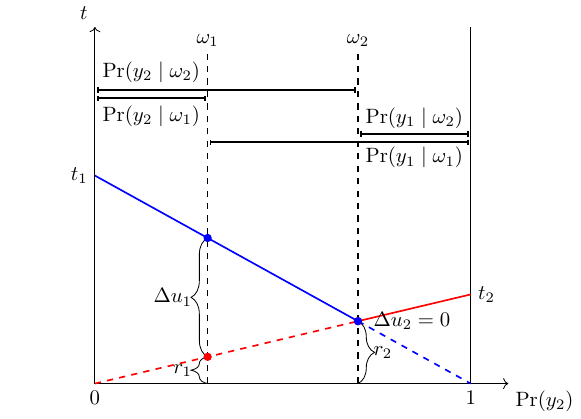}
\par\end{center}%
\end{minipage}\quad{}%
\begin{minipage}[t]{0.48\columnwidth}%
\begin{center}
{\small Panel B: Provide $(0,\Delta u_{2})$}{\small\par}
\par\end{center}
\begin{center}
\includegraphics[width=1\columnwidth]{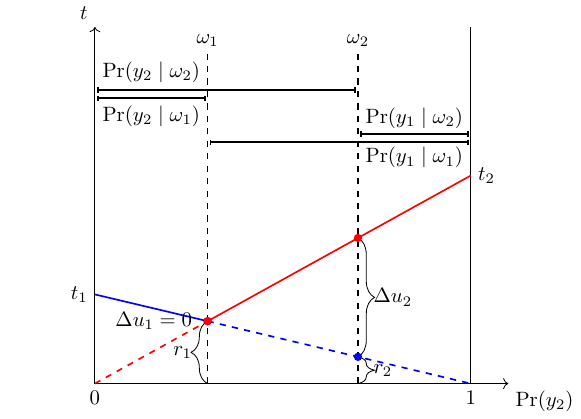}
\par\end{center}%
\end{minipage}
\par\end{centering}
\centering{}%
\begin{minipage}[t]{0.9\columnwidth}%
\begin{singlespace}
{\scriptsize Notes: This figure illustrates the relation between the
rents $r_{1}$ and $r_{2}$ and the incentives $\Delta u_{1}$ and
$\Delta u_{2}$. Panel A plots the contract that provides incentive
profile $\left(\Delta u_{1},0\right)$, and Panel B plots the contract
that provides incentive profile $\left(0,\Delta u_{2}\right)$.}{\scriptsize\par}
\end{singlespace}

\end{minipage}
\end{figure}

\section{Conclusion\label{sec:conclusion}}

I conclude by discussing paths for future research. First, one can
consider the profit maximization phase now that we understand the
cost minimization phase of moral hazard models. If the principal can
endogenously choose what information the agent should acquire or what
outcome distribution the agent should produce, how to maximize the
principal's profits? Second, one can also endogenize the principal's
choice of $\mathcal{E}_{P}$ to study the trade-off between the information
costs and the agency costs. In this case, one can study how much the
principal should learn to maximize her profits. Third, my framework
is useful to also study the principal's credibility. I assume that
$\mathcal{E}_{P}$ is publicly known with its realization contractible
throughout the paper, but the principal may have an incentive to misreport
$\mathcal{E}_{P}$'s realization. Lastly, my framework can also be
extended to study the case with multiple agents. 
\pagebreak{}

\bibliographystyle{format/aea/aea}
\bibliography{paper}

\pagebreak{}

\appendix
\setcounter{page}{1}
\section{Omitted Proofs \label{appsec:proof-appendix}}

I first present a theorem due to \citet{baksalary1979matrix} on matrix
equations that proves to be useful later in this appendix. 
\begin{thm}
[Baksalary and Kala (1979)]\label{thm:AX-YB=00003DC} Suppose $A,B,C$
are constant matrices and $X,Y$ are unknowns. All matrices have conformable
shapes. Matrix equation $AX-YB=C$ has a solution if and only if $(I-AA^{-})C(I-B^{-}B)=\boldsymbol{0}$
where $A^{-},B^{-}$are any generalized inverse of $A,B$, respectively.
Moreover, the general solution, if one exists, is given by 
\begin{align*}
X & =A^{-}C+A^{-}ZB+(I-A^{-}A)W,\\
Y & =-(I-AA^{-})CB^{-}+Z-(I-AA^{-})ZBB^{-},
\end{align*}
where $W$ and $Z$ are arbitrary matrices of conformable shapes,
and $I$ is the identity matrix of conformable shapes.
\end{thm}
\begin{rem}
In the main text, I use $\mathcal{E}_{P}^{-}$ to represent its Moore-Penrose
pseudo-inverse because it has nice interpretation of the projection
operator, while in Theorem \ref{thm:AX-YB=00003DC}, $A^{-}$ and
$B^{-}$ can be any generalized inverse. 
\end{rem}
We are now ready for the proofs.

\subsection{Omitted Proofs in Section \ref{sec:implementability}}
\begin{proof}
[Proof of Lemma \ref{lem:imp-matrix-eqn}] Recall the agent's problem
is to choose a posterior distribution to maximize a posterior separable
objective (\ref{eq:A-problem-post-sep}) subject to the Bayes plausibility
constraint (\ref{eq:Bayes-plausibility}). Let $\lambda\in\mathbb{R}^{N}$
be the multiplier on (\ref{eq:Bayes-plausibility}), and the problem
becomes
\[
\max_{\mathcal{E}\in E^{K}}\:\mathbb{E}_{\mu_{k}\sim\mathcal{E}}\left[\mu_{k}\cdot\mathcal{E}_{P}T_{k}-c(\mu_{k})+\lambda\cdot\left(\mu_{0}-\mu_{k}\right)\right].
\]
If some $\mathcal{E}^{*}$ is optimal, its induced posterior belief
$\mu_{k}$ at which the agent reports $x_{k}$ must also be optimal
in the following sense, 
\[
\mu_{k}\in\underset{\mu\in\Delta\Omega}{\arg\max}\:\mu\cdot\mathcal{E}_{P}T_{k}-c(\mu)+\lambda\cdot\left(\mu_{0}-\mu\right).
\]
The first order condition is 
\[
\mathcal{E}_{P}T_{k}=\nabla_{k}+\lambda,\forall1\leq k\leq K,
\]
which becomes (\ref{eq:imp-matrix-eqn}) when stacked together. Assumptions
\ref{assu:tech-post-sep} and \ref{assu:smooth-post-sep} imply that
the derivatives are well defined, and that the first order condition
also suffices for optimality. Assumption \ref{assu:inf-slope-post-sep}
implies that this first order condition must be satisfied as equality.\footnote{Without assumption (\ref{assu:inf-slope-post-sep}), the agent may
choose to have a non-fully-supported belief with $\mu_{k}^{n}=0$
for some report $x_{k}$ and state $\omega_{n}$. When that happens,
there are additional multipliers on constraint $\mu_{k}^{n}\geq0$
and the above first order condition becomes inequality. Assumption
(\ref{assu:inf-slope-post-sep}) rules out these corner cases as it
is infinitely costly to rule out a state.} 

To finish the proof, suppose a contract $T$ under $\mathcal{E}_{P}$
implements $\mathcal{E}_{A}$, it must satisfy the first order condition
(\ref{eq:imp-matrix-eqn}) for some $\lambda\in\mathbb{R}^{N}$. Conversely,
if (\ref{eq:imp-matrix-eqn}) holds for some $T$ and $\lambda$,
by the sufficiency of the first order condition, $T$ implements $\mathcal{E}_{A}$.
\end{proof}
\begin{proof}
[Proof of Theorem \ref{thm:imp}] By Lemma \ref{lem:imp-matrix-eqn},
the implementability problem is converted into the solvability problem
of (\ref{eq:imp-matrix-eqn}). Applying Theorem \ref{thm:AX-YB=00003DC}
due to \citet{baksalary1979matrix}, the solvability of (\ref{eq:imp-matrix-eqn})
is equivalent to the solvability of the following equation, 
\begin{equation}
\left(I_{N}-\mathcal{E}_{P}\mathcal{E}_{P}^{-}\right)\nabla\left(I_{K}-\dfrac{1}{K}\bm{1}_{K\times K}\right)=\bm{0}_{N\times K},\label{eq:solvability}
\end{equation}
where $I_{N},I_{K}$ are identity matrices of sizes $N$ and $K$,
and $\mathcal{E}_{P}^{-}$ is the Moore-Penrose pseudo-inverse of
matrix $\mathcal{E}_{P}$. 

The last step is to obtain condition (\ref{eq:imp-cond}) by observing
that $\mathcal{E}_{P}\mathcal{E}_{P}^{-}$ is the orthogonal projection
matrix onto the column space of $\mathcal{E}_{P}$, and the second
parenthesis in (\ref{eq:solvability}) represents taking the column-wise
differences of $\nabla$. I provide a detailed argument as follows.
First, apply a invertible linear transformation on (\ref{eq:solvability}).
Multiplying (\ref{eq:solvability}) from the right by any invertible
matrix yields an equivalent matrix equation. I can then multiply (\ref{eq:solvability})
from the right by the following $K\times K$ matrix
\[
Q=\begin{bmatrix}1 & 0 & \dots & 0 & \frac{1}{K}\\
0 & 1 & \dots & 0 & \frac{1}{K}\\
\vdots & \vdots & \ddots & \vdots & \vdots\\
0 & 0 & \dots & 1 & \frac{1}{K}\\
-1 & -1 & \dots & -1 & \frac{1}{K}
\end{bmatrix},
\]
and (\ref{eq:solvability}) becomes 
\begin{equation}
\left(I_{N}-\mathcal{E}_{P}\mathcal{E}_{P}^{-}\right)\nabla Q'=\bm{0}_{N\times K}\label{eq:solvability-transformed}
\end{equation}
where $Q'$ is a $K\times K$ matrix given by 
\[
Q'=\begin{bmatrix}1 & 0 & \dots & 0 & 0\\
0 & 1 & \dots & 0 & 0\\
\vdots & \vdots & \ddots & \vdots & \vdots\\
0 & 0 & \dots & 1 & 0\\
-1 & -1 & \dots & -1 & 0
\end{bmatrix}.
\]
The interpretation of multiplying $Q'$ to the right of $\nabla$
is to take the difference between any column of $\nabla$ and its
last column. 

Next, by the property of pseudo-inverses, $\mathcal{E}_{P}\mathcal{E}_{P}^{-}$
is the orthogonal projection matrix onto the column space of $\mathcal{E}_{P}$.
Therefore, (\ref{eq:solvability-transformed}) says that columns of
$\nabla Q'$ are in the column space of $\mathcal{E}_{P}$. That is,
the implementability of $\mathcal{E}_{A}$ is equivalent to 
\[
\nabla_{k}-\nabla_{K}\in\operatorname{Col}\mathcal{E}_{P},\forall1\leq k\leq K,
\]
where $\nabla_{K}$ is the last column of $\nabla$. By the properties
of linear subspaces, the above condition is equivalent to Condition
(\ref{eq:imp-cond}) in Theorem \ref{thm:imp}, thus completing the
proof. 
\end{proof}
\begin{proof}
[Proof of Corollary \ref{cor:full-imp}] Start from the agent's problem.
Without Assumption \ref{assu:inf-slope-post-sep}, we need to explicitly
take care of the constraint that the probabilities chosen in the agent's
problem (\ref{eq:A-problem-post-sep}) must be positive. To this end,
we need to apply Theorem \ref{thm:imp-corner} in Appendix \ref{appsec:imp-corner-solutions}.
It is clear from Condition (\ref{eq:imp-cond-corner}) that whenever
$\mathcal{E}_{P}$ has full row rank, any feasible $\mathcal{E}_{A}$
is implementable. 

\end{proof}
\begin{proof}
[Proof of Corollary \ref{cor:binary-state}] When $N=2$, $\mathcal{E}_{P}$
has only two rows. If $\mathcal{E}_{P}$ does not have full row rank,
linear dependency implies the two rows of $\mathcal{E}_{P}$ must
be multiples of each other. Row sums being one implies they must be
the same. Such $\mathcal{E}_{P}$ cannot be informative. Therefore,
any informative $\mathcal{E}_{P}$ has full row rank. (\ref{eq:solvability})
is equivalent to 
\end{proof}
\begin{proof}
[Proof of Corollary \ref{cor:random-noise}] Given an $\mathcal{E}_{P}$
that has uniform random noise. It suffices to show that $\mathcal{E}_{P}$
has full row rank. By Definition \ref{def:random-noise}, $\mathcal{E}_{P}$
is a rectangular stochastic matrix with $N$ rows where each row has
a unique largest entry, none of these row-wise largest entry shares
the same column, and given any state the conditional probabilities
for other realizations are equal. Take out the $N$ rows of $\mathcal{E}_{P}$
each with a row-wise largest entries, and form a square submatrix
$\tilde{\mathcal{E}}_{P}$.\footnote{The submatrix $\tilde{\mathcal{E}}_{P}$ may no longer be an experiment
as the rows may not sum to one.} Reindex the columns of $\tilde{\mathcal{E}}_{P}$ so that the maximum
entry in the $n$-th row appears in the $n$-th column. If $\tilde{\mathcal{E}}_{P}$
is invertible, it has full row rank, which then implies that $\mathcal{E}_{P}$
has full row rank. Since multiplying a non-zero constant on each row
does not change invertibility, we can normalize the conditional probabilities
of entries that are not the row-wise maximum to one. This way, the
resulting matrix must take the following form, 
\[
\begin{bmatrix}a_{1} & 1 & \cdots & 1\\
1 & a_{2} & \cdots & 1\\
\vdots & \vdots & \ddots & \vdots\\
1 & 1 & \cdots & a_{N}
\end{bmatrix}=\operatorname{diag}\left(a_{1},a_{2},...,a_{N}\right)+\bm{1}_{N\times N}
\]
with $a_{1},a_{2},...,a_{n}>1$. Apply the Sherman--Morrison theorem
for matrix inversion, the above matrix is invertible if and only if
\[
1+\bm{1}_{1\times N}\operatorname{diag}\left(a_{1}^{-1},a_{2}^{-1},...,a_{N}^{-1}\right)\bm{1}_{N\times1}\neq0,
\]
which is obviously true.
\end{proof}
\begin{proof}
[Proof of Proposition \ref{prop:imp-comparison}] Suppose $\operatorname{Col}\mathcal{E}_{P}\supseteq\operatorname{Col}\mathcal{E}_{P}'$.
Pick any $C\in\mathcal{C}$ and $\mathcal{E}_{A}\in E$. Theorem \ref{thm:imp}
immediately implies that if $\mathcal{E}_{A}$ is implementable under
$\mathcal{E}_{P}'$, it is also implementable under $\mathcal{E}_{P}$.
Therefore, $\mathcal{I}^{C,\mu_{0}}(\mathcal{E}_{P})\supseteq\mathcal{I}^{C,\mu_{0}}(\mathcal{E}_{P}')$.

Suppose $\mathcal{I}^{C,\mu_{0}}(\mathcal{E}_{P})\supseteq\mathcal{I}^{C,\mu_{0}}(\mathcal{E}_{P}')$.
We show that $\operatorname{Col}\mathcal{E}_{P}\supseteq\operatorname{Col}\mathcal{E}_{P}'$.
Suppose to the contrary that $\operatorname{Col}\mathcal{E}_{P}\not\supseteq\operatorname{Col}\mathcal{E}_{P}'$.
Claim \ref{claim:v-not-all-pos-neg} below says that there exists
some $v\in\operatorname{Col}\mathcal{E}_{P}'\setminus\operatorname{Col}\mathcal{E}_{P}$
such that $v\not\geq\bm{0}$ and $v\not\leq\bm{0}$. The proof of
this claim is deferred because it is not related to the main argument.
Consider a matrix $\nabla\in\mathbb{R}^{N\times2}$ with $\nabla_{1}-\nabla_{2}=v$.
Since $\nabla$ satisfies the no dominance condition, Lemma \ref{lem:nabla-no-dominance}
implies that there exists some $\mathcal{E}_{A}\in E$ and $C\in\tilde{C}$
such that $\nabla$ is the marginal cost matrix. This $\mathcal{E}_{A}$
is implementable under $\mathcal{E}_{P}'$ but not $\mathcal{E}_{P}$,
contradicting $\mathcal{I}^{C,\mu_{0}}(\mathcal{E}_{P})\supseteq\mathcal{I}^{C,\mu_{0}}(\mathcal{E}_{P}')$. 

Suppose $\operatorname{Col}\mathcal{E}_{P}\supsetneq\operatorname{Col}\mathcal{E}_{P}'$.
The same construction as above then implies that $\mathcal{I}^{C,\mu_{0}}(\mathcal{E}_{P})\supsetneq\mathcal{I}^{C,\mu_{0}}(\mathcal{E}_{P}')$. 

Suppose $\mathcal{I}^{C,\mu_{0}}(\mathcal{E}_{P})\supsetneq\mathcal{I}^{C,\mu_{0}}(\mathcal{E}_{P}')$.
We immediately have $\operatorname{Col}\mathcal{E}_{P}\supseteq\operatorname{Col}\mathcal{E}_{P}'$.
If $\operatorname{Col}\mathcal{E}_{P}=\operatorname{Col}\mathcal{E}_{P}'$,
then Theorem \ref{thm:imp} implies that $\mathcal{I}^{C,\mu_{0}}(\mathcal{E}_{P})=\mathcal{I}^{C,\mu_{0}}(\mathcal{E}_{P}')$.
Therefore, we must have $\operatorname{Col}\mathcal{E}_{P}\supsetneq\operatorname{Col}\mathcal{E}_{P}'$.
\end{proof}
\begin{claim}
\label{claim:v-not-all-pos-neg}Suppose $\operatorname{Col}\mathcal{E}_{P}\not\supseteq\operatorname{Col}\mathcal{E}_{P}'$.
There exists some $v\in\operatorname{Col}\mathcal{E}_{P}'\setminus\operatorname{Col}\mathcal{E}_{P}$
such that $v\not\geq0$ and $v\not\leq0$.
\end{claim}
\begin{proof}
Since $\operatorname{Col}\mathcal{E}_{P}\not\supseteq\operatorname{Col}\mathcal{E}_{P}'$,
there is some $\tilde{v}\in\operatorname{Col}\mathcal{E}_{P}'\setminus\operatorname{Col}\mathcal{E}_{P}$.
I can always add to $\tilde{v}$ a scalar $\alpha<0$ times a column
$\mathcal{E}_{k}'$ of $\mathcal{E}_{P}'$ with $\mathcal{E}_{k}'\not\in\operatorname{Col}\mathcal{E}_{P}$.
Such a column must exist due to $\operatorname{Col}\mathcal{E}_{P}\not\supseteq\operatorname{Col}\mathcal{E}_{P}'$.
The resulting vector $v=\tilde{v}+\alpha\mathcal{E}_{k}'$ satisfies
$v\in\operatorname{Col}\mathcal{E}_{P}'$ and $v\not\geq0$. It remains
to show that I can do this without making $v\in\operatorname{Col}\mathcal{E}_{P}$
or $v\leq0$.

First, $v\in\operatorname{Col}\mathcal{E}_{P}$ can be avoided. There
is at most one $\alpha$ that can make $\tilde{v}+\alpha\mathcal{E}_{k}'\in\operatorname{Col}\mathcal{E}_{P}$,
since if two or more such $\alpha$ values would imply $\mathcal{E}_{k}'\in\operatorname{Col}\mathcal{E}_{P}$.
Therefore, if $v\in\operatorname{Col}\mathcal{E}_{P}$, simply perturb
$\alpha$ slightly. 

Second, $v\leq0$ can be avoided. If the resulting $v\leq0$, we can
reduce $|\alpha|$ to make some entries positive. The only case that
requires extra care is that as $|\alpha|$ increases, $v$ goes from
$v\geq0$ to $v\leq0$ directly. This happens only if $\tilde{v}$
is parallel to $\mathcal{E}_{k}'$. In this case, simply pick a different
$\mathcal{E}_{k}''$ to construct $v$. There exists another column
$\mathcal{E}_{k}''$ of $\mathcal{E}_{P}'$ with $\mathcal{E}_{k}''\not\in\operatorname{Col}\mathcal{E}_{P}$
because of the row sums being one. More specifically, if $\mathcal{E}_{k}'$
is the only column of $\mathcal{E}_{P}'$ with $\mathcal{E}_{k}'\not\in\operatorname{Col}\mathcal{E}_{P}$,
it can be expressed as a linear combination of all other columns of
$\mathcal{E}_{P}'$, which are in $\operatorname{Col}\mathcal{E}_{P}$,
and the vector of ones $\bm{1}$, which is also in $\operatorname{Col}\mathcal{E}_{P}$
since $\mathcal{E}_{P}\bm{1}=\bm{1}$. This completes the proof. 
\end{proof}

\subsection{Omitted Proofs in Section \ref{sec:IR-cost-min}}

I will first prove Theorem \ref{thm:imp-contract-decomposition}.
Next, I will prove several useful lemmas before proving other results
in Section \ref{sec:IR-cost-min}. 
\begin{proof}
[Proof of Theorem \ref{thm:imp-contract-decomposition}] Suppose some
contract $t$ under $\mathcal{E}_{P}$ implements $\mathcal{E}_{A}$.
$t$ must satisfy (\ref{eq:imp-matrix-eqn}) for some $\lambda\in\mathbb{R}^{N}$.
Applying \citet{baksalary1979matrix}'s general solution formulas
(see Theorem \ref{thm:AX-YB=00003DC}), we have 
\begin{align}
T & =\mathcal{E}_{P}^{-}\nabla+\mathcal{E}_{P}^{-}\tilde{Z}\boldsymbol{1}_{1\times K}+\left(I_{M}-\mathcal{E}_{P}^{-}\mathcal{E}_{P}\right)\tilde{W}\label{eq:t-general-sol}\\
\lambda & =-\dfrac{1}{K}\left(I_{N}-\mathcal{E}_{P}\mathcal{E}_{P}^{-}\right)\nabla\boldsymbol{1}_{K\times1}+\mathcal{E}_{P}\mathcal{E}_{P}^{-}\tilde{Z}\label{eq:lambda-general-sol}
\end{align}
where $\tilde{Z}\in\mathbb{R}^{N}$ and $\tilde{W}\in\mathbb{R}^{M\times K}$
are free variables. Let $Z=\mathcal{E}_{P}^{-}\tilde{Z}$ and $W=\left(I_{M}-\mathcal{E}_{P}^{-}\mathcal{E}_{P}\right)\tilde{W}$.
(\ref{eq:t-general-sol}) becomes 
\[
T=\mathcal{E}_{P}^{-}\nabla+Z\boldsymbol{1}_{1\times K}+W.
\]
It suffices to show that $Z$ can take any values and $W$ can take
any values as long as $\mathcal{E}_{P}W=\bm{0}$. For the former,
since $Z=\mathcal{E}_{P}^{-}\tilde{Z}\in\operatorname{Col}\mathcal{E}_{P}^{-}=\operatorname{Row}\mathcal{E}_{P}$
and $\tilde{Z}$ is a free variable, $Z$ can be any vector in $\operatorname{Row}\mathcal{E}_{P}$.
However, since $\mathbb{R}^{M}$ can be orthogonal decomposed into
$\operatorname{Row}\mathcal{E}_{P}$ and $\operatorname{Ker}\mathcal{E}_{P}$,
it is without loss for $Z$ to take any value in $\mathbb{R}^{M}$.
This is because we can always uniquely decompose $Z=Z^{\parallel}+Z^{\perp}$
where $Z^{\parallel}\in\operatorname{Row}\mathcal{E}_{P}$ and $Z^{\perp}\in\operatorname{Ker}\mathcal{E}_{P}$,
and then absorb $Z^{\perp}$ into $W$. For the latter, $\mathcal{E}_{P}^{-}\mathcal{E}_{P}$
is the orthogonal projection matrix onto $\operatorname{Col}\mathcal{E}_{P}$.
Therefore, by definition,
\[
\mathcal{E}_{P}W=\mathcal{E}_{P}\left(I_{M}-\mathcal{E}_{P}^{-}\mathcal{E}_{P}\right)\tilde{W}=\bm{0},
\]
and since $\tilde{W}$ is free, $W$ can be anything whose columns
are in the null space of the orthogonal projection operator onto $\operatorname{Col}\mathcal{E}_{P}$.
\end{proof}
The next lemma uses Theorem \ref{thm:imp-contract-decomposition}
to rewrite the principal's problem (\ref{eq:P-problem}).
\begin{lem}
\label{lem:P-problem-matrix} The principal's problem (\ref{eq:P-problem})
is equivalent to 
\begin{align*}
\min_{Z,W_{k}\in\mathbb{R}^{M}} & \:\mu_{0}^{\intercal}\mathcal{E}_{P}\mathcal{E}_{P}^{-}\nabla\mathcal{E}_{A}^{\intercal}\mu_{0}+\mu_{0}^{\intercal}\mathcal{E}_{P}Z,\\
\text{s.t. } & \mathcal{E}_{P}W_{k}=\boldsymbol{0},\forall1\leq k\leq K,\\
 & \mathcal{E}_{P}^{-}\nabla_{k}+Z+W_{k}\geq\boldsymbol{0},\forall1\leq k\leq K.
\end{align*}
\end{lem}
\begin{proof}
Due to Theorem \ref{thm:imp-contract-decomposition}, the incentive
constraint in the principal's problem (\ref{eq:P-problem}) is equivalent
to (\ref{eq:imp-payment-rule}). Plug (\ref{eq:P-problem}) into the
objective and the constraints, note that the participation constraint
does not bind due to limited liability, and we obtain the expression
above. 
\end{proof}

I now present a lemma that identifies a no dominance condition under
which a matrix $\nabla$ can be the marginal cost matrix of some experiment
$\mathcal{E}_{A}$ under some cost function $C\in\mathcal{C}$. I
first state the no dominance condition. In words, the no dominance
condition says that, if we view $\nabla$ as a matrix of state-dependent
utility with rows representing the states and columns representing
the actions, no action is weakly dominated. The definition allows
for identical columns in $\nabla_{k}$. I will identify $\nabla$
as a collection of subgradients of some convex function $c$, and
allowing for identical columns handles the case where $c$ has an
affine part.
\begin{defn}
[No Dominance] Given a matrix $\nabla\in\mathbb{R}^{N\times K}$
with its columns denoted by $\nabla_{k}$. Say that a column $\nabla_{\tilde{k}}$
is dominated if there exists scalars $\alpha_{k}\geq0$ for $k\neq\tilde{k}$
such that $\sum_{k\neq\tilde{k}}\alpha_{k}=1$ and $\sum_{k\neq\tilde{k}}\alpha_{k}\nabla_{k}\geq\nabla_{\tilde{k}}$.
Say that $\nabla$ satisfies the no dominance condition if no column
is dominated, unless it is weakly dominated by another column identical
to it. 
\end{defn}
The following lemma relates matrices satisfying the no dominance condition
to the marginal cost matrices. This comes from the observation that
any marginal cost matrix is a collection of subgradient of some differentiable
convex function, which must satisfy the no dominance condition.
\begin{lem}
\label{lem:nabla-no-dominance} A matrix $\nabla\in\mathbb{R}^{N\times K}$
satisfies the no dominance condition if and only if there exists some
constant $b\in\mathbb{R}$ such that $\nabla+b$ is the marginal cost
matrix for some experiment $\mathcal{E}_{A}\in E$ under some cost
function $C\in\mathcal{C}$.\footnote{$\nabla+b$ stands for adding $b$ to each element of $\nabla$. I
need the constant $b$ so that the posterior cost function $c$ can
satisfy $c(\mu_{0})=0$.}
\end{lem}
\begin{proof}
[Proof of Lemma \ref{lem:nabla-no-dominance}] Suppose for some $b\in\mathbb{R}$,
$\nabla+b$ is a marginal cost matrix for some $\mathcal{E}_{A}\in E$
under some $C\in\mathcal{C}$ with posterior cost $c$. Let $\tilde{c}(\mu):=c(\mu)+b$.
By definition, $\nabla$ must be a subgradient matrix of the differentiable
convex function $\tilde{c}$. Specifically, each column $\nabla_{k}$
must be the unique subgradient of $\tilde{c}$ at some $\mu_{k}$.
The uniqueness of the subgradient comes from differentiability. It
represents a supporting hyperplane in the space of $(\mu,y)$ by $\mu\cdot\nabla_{k}=y-b$.\footnote{The hyperplane is $\left(\mu_{k}-\mu\right)\cdot\nabla_{k}+y-\tilde{c}(\mu_{k})=0$,
which simplifies to the above equation due to the normalization in
Definition (\ref{def:MC-post-sep}) which says $\mu_{k}\cdot\nabla_{k}=c(\mu_{k})$} If there is any column $\nabla_{k}$ dominated by a convex combination
of columns that are not identical to it, then the supporting hyperplane
$\mu\cdot\nabla_{k}=y-b$ at $\mu_{k}$ lies everywhere below the
upper envelope of the supporting hyperplanes at $\mu_{k'}$ with $k'\neq k$,
which means $\nabla_{k}$ cannot be a subgradient. 

Conversely, suppose $\nabla$ is a matrix that satisfies the no dominance
condition. Let $\tilde{\nabla}$ be a matrix constructed by the distinct
columns of $\nabla$. The fundamental theorem of convex analysis implies
that $\tilde{\nabla}$ is a subgradient matrix of some differentiable
convex function $c$. Pick some constant $b\in\mathbb{R}$ so that
$c(\mu_{0})+b=0$, and define $\tilde{c}(\mu):=c(\mu)+b$. $\tilde{c}$
induces a posterior separable cost $\tilde{C}\in\mathcal{C}$. Moreover,
each column $\tilde{\nabla}_{k}$ is the unique subgradient of $\tilde{c}$
at distinct $\mu_{k}$. As a result, there exists a posterior distribution
$\mathcal{E}_{A}$ with the marginal cost matrix $\tilde{\nabla}$
under $\tilde{C}$. We can always construct $\mathcal{E}_{A}$ to
be Bayes plausible because the posteriors $\mu_{k}$ can always be
chosen to include the prior $\mu_{0}$ in their convex hull. We now
add back duplicate columns by creating duplicate realizations in $\mathcal{E}_{A}$.
This completes the proof. 
\end{proof}
\begin{prop}
[Without Limited Liability]\label{prop:EAIR-opt-contract} Without
limited liability, any $\mathcal{E}_{A}$ implementable under $\mathcal{E}_{P}$
can be implemented at a cost of $C(\mathcal{E}_{A})$. 
\end{prop}
\begin{proof}
[Proof of Proposition \ref{prop:EAIR-opt-contract}] It suffices to
find a contract that implements $\mathcal{E}_{A}$ at a cost of $C(\mathcal{E}_{A})$.
As discussed in the main text, we can start from any contract that
implements $\mathcal{E}_{A}$ and then make the participation constraint
bind. 

The following contract $T$ is a valid construction,
\begin{equation}
T=\mathcal{E}_{P}^{-}\nabla+z\mathcal{E}_{P}^{-}\bm{1}_{M\times K}\text{ with }z=\dfrac{C(\mathcal{E}_{A})-\mu_{0}^{\intercal}\mathcal{E}_{P}\mathcal{E}_{P}^{-}\nabla\mathcal{E}_{A}^{\intercal}\mu_{0}}{\mu_{0}^{\intercal}\mathcal{E}_{P}\mathcal{E}_{P}^{-}\boldsymbol{1}_{M\times1}}.\label{eq:EAIR-contract}
\end{equation}

I first show that (\ref{eq:EAIR-contract}) implements $\mathcal{E}_{A}$.
(\ref{eq:EAIR-contract}) is obtained from (\ref{eq:imp-payment-rule})
by setting $W_{k}=\bm{0}$ and $Z=z\mathcal{E}_{P}^{-}\bm{1}_{M\times1}$.
Theorem \ref{thm:imp-contract-decomposition} says that (\ref{eq:EAIR-contract})
such a $T$ implements $\mathcal{E}_{A}$.

Next, I show that (\ref{eq:EAIR-contract}) makes the participation
constraint bind. Using Lemma (\ref{lem:P-problem-matrix}), the principal's
objective is
\begin{align*}
\mathbb{E}_{\mu\sim\left\langle \mathcal{E}\mid\mu_{0}\right\rangle ,y_{m}\sim\left\langle \mathcal{E}_{P}\mid\mu\right\rangle }\left[T(y_{m},s(\mu))\right] & =\mu_{0}^{\intercal}\mathcal{E}_{P}\mathcal{E}_{P}^{-}\nabla\mathcal{E}_{A}^{\intercal}\mu_{0}+z\mu_{0}^{\intercal}\mathcal{E}_{P}\mathcal{E}_{P}^{-}\bm{1}_{M\times K}\mathcal{E}_{A}^{\intercal}\mu_{0}\\
 & =\mu_{0}^{\intercal}\mathcal{E}_{P}\mathcal{E}_{P}^{-}\nabla\mathcal{E}_{A}^{\intercal}\mu_{0}+z\mu_{0}^{\intercal}\mathcal{E}_{P}\mathcal{E}_{P}^{-}\bm{1}_{M\times1}\\
 & =\mu_{0}^{\intercal}\mathcal{E}_{P}\mathcal{E}_{P}^{-}\nabla\mathcal{E}_{A}^{\intercal}\mu_{0}+\dfrac{C(\mathcal{E}_{A})-\mu_{0}^{\intercal}\mathcal{E}_{P}\mathcal{E}_{P}^{-}\nabla\mathcal{E}_{A}^{\intercal}\mu_{0}}{\mu_{0}^{\intercal}\mathcal{E}_{P}\mathcal{E}_{P}^{-}\boldsymbol{1}_{N\times1}}\mu_{0}^{\intercal}\mathcal{E}_{P}\mathcal{E}_{P}^{-}\bm{1}_{M\times1}\\
 & =C(\mathcal{E}_{A}).
\end{align*}
\end{proof}

\begin{proof}
[Proof of Proposition \ref{prop:LL-opt-contract}] When $\mathcal{E}_{P}$
has full row rank, we must have $W_{k}=\boldsymbol{0}$. Using Lemma
(\ref{lem:P-problem-matrix}), the optimal solution is simply to make
the limited liability constraint bind for every realization of $\mathcal{E}_{P}$.
Therefore, the contract in (\ref{prop:LL-opt-contract}) is optimal. 
\end{proof}

\begin{proof}
[Proof of Proposition \ref{prop:LL-cost-comparison}] I first rewrite
the principal's cost minimization problem in terms of the agent's
interim utilities assuming $\mathcal{E}_{A}$ is implementable under
$\mathcal{E}_{P}$. Let $u_{k}=\mathcal{E}_{P}\mathcal{E}_{P}^{-}\nabla_{k}+\mathcal{E}_{P}\tilde{Z}\in\mathbb{R}^{N}$
be the agent's state-dependent utility from reporting $x_{k}$. Observe
that 
\[
u_{k}-u_{1}=\mathcal{E}_{P}\mathcal{E}_{P}^{-}\left(\nabla_{k}-\nabla_{1}\right)=\nabla_{k}-\nabla_{1},
\]
where the last equality follows from the implementability condition
(\ref{eq:imp-payment-rule}) and that $\mathcal{E}_{P}\mathcal{E}_{P}^{-}$
is the orthogonal projection matrix onto $\operatorname{Col}\mathcal{E}_{P}$.
The principal's objective can then be rewritten in terms of $u_{1}$,
\begin{align*}
P_{Y}^{\intercal}TP_{X} & =\mu_{0}^{\intercal}\left(\mathcal{E}_{P}\mathcal{E}_{P}^{-}\nabla+\mathcal{E}_{P}\tilde{Z}\bm{1}_{1\times K}\right)\mathcal{E}_{A}^{\intercal}\mu_{0}\\
 & =\mu_{0}^{\intercal}u_{1}+\sum_{k=1}^{K}P_{X}^{k}\mu_{0}^{\intercal}\left(\nabla_{k}-\nabla_{1}\right)
\end{align*}
where $P_{X}^{k}:=\Pr(x_{k})$ is the unconditional probability of
$\mathcal{E}_{A}$ sending $x_{k}$. The objective only involves $u_{1}$
and its second term is a constant. Similarly, the IIR constraint becomes
$u_{k}=u_{1}+\nabla_{k}-\nabla_{1}\geq\boldsymbol{0},\forall1\leq k\leq K$
which can be summarized as $u_{1}\geq\underline{u}$ where $\underline{u}^{n}:=\max_{k}\left\{ \nabla_{1}^{n}-\nabla_{k}^{n}\right\} $.
The requirement of $\tilde{Z}\in\operatorname{Row}\mathcal{E}_{P}$
becomes $u_{1}\in\operatorname{Col}\mathcal{E}_{P}$ since $\tilde{Z}$
can be any vector in the row space of $\mathcal{E}_{P}$. Rewrite
the objective and the EPIR constraints with $u_{1}$, and the principal's
problem becomes
\begin{align}
\min_{u_{1}\in\mathbb{R}^{N},\tilde{W}\in\mathbb{R}^{M\times K}} & \:\mu_{0}^{\intercal}u_{1}\label{eq:EPIR-cost-min-u}\\
\text{s.t. } & u_{1}\in\operatorname{Col}\mathcal{E}_{P}\nonumber \\
 & \mathcal{E}_{P}\tilde{W}=\boldsymbol{0},\nonumber \\
 & \mathcal{E}_{P}^{-}\left(u_{1}+\nabla_{k}-\nabla_{1}\right)+\tilde{W}_{k}\geq0,\forall1\leq k\leq K,\nonumber 
\end{align}
where I omit a constant that only depends on $\mathcal{E}_{A}$ and
$C$ in the objective. This constant is irrelevant for the princpal's
decision problem, nor is it relevant for comparing the implementation
costs under different $\mathcal{E}_{P}$'s.

We are now ready to prove the proposition. Suppose $\mathcal{E}_{P}\in E^{M}$
and $\mathcal{E}_{P}'\in E^{M'}$ satisfy $\mathcal{E}_{P}\geq_{B,2}\mathcal{E}_{P}'$.
Then $\mathcal{E}_{P}'=\mathcal{E}_{P}G$ for some matrix $G\in\mathbb{R}^{M\times M'}$
with $G\geq0$. This immediately implies that $\operatorname{Col}\mathcal{E}_{P}\supseteq\operatorname{Col}\mathcal{E}_{P}'$.
Take any $C\in\mathcal{C}$ and $\mathcal{E}_{A}\in E$. Let $\nabla$
be the marginal cost matrix. If $\mathcal{E}_{A}$ is not implementable
under $\mathcal{E}_{P}'$, there is nothing to prove since $\kappa^{C,\mu_{0}}(\mathcal{E}_{A},\mathcal{E}_{P}')=+\infty$.
If $\mathcal{E}_{P}$ is implementable under $\mathcal{E}_{P}'$,
it must be also implementable under $\mathcal{E}_{P}$ due to the
column space inclusion. This means we can rewrite the principal's
problems under $\mathcal{E}_{P}$ and $\mathcal{E}_{P}'$ as (\ref{eq:EPIR-cost-min-u}).
In this case, it suffices to focus on what $u_{1}$ is feasible in
(\ref{eq:EPIR-cost-min-u}) since the principal's payoff depends only
on it. Say that a $u_{1}$ is feasible under $\mathcal{E}_{P}$ if
there exists some $\tilde{W}$ such that $\left(u_{1},\tilde{W}\right)$
satisfies the constraints in problem (\ref{eq:EPIR-cost-min-u}).
To show that the EPIR cost is always lower under $\mathcal{E}_{P}$,
it suffices to show that any $u_{1}'$ that is feasible under $\mathcal{E}_{P}'$
is also feasible under $\mathcal{E}_{P}$. Take any $u_{1}$ that
is feasible under $\mathcal{E}_{P}'$. We know that $u_{1}\in\operatorname{Col}\mathcal{E}_{P}G$
and that there exists $\tilde{W}'$ such that the constraints in (\ref{eq:EPIR-cost-min-u})
are satisfied. I have to show that $u_{1}$ is also feasible under
$\mathcal{E}_{P}$. Obviously, $u_{1}\in\operatorname{Col}\mathcal{E}_{P}$.
Let $\tilde{W}$ be defined by
\[
\tilde{W}_{k}:=G\tilde{W}_{k}'+\left(GG^{-}-I\right)\mathcal{E}_{P}^{-}\left(u_{1}+\nabla_{k}-\nabla_{1}\right),\forall1\leq k\leq K.
\]
I now show that this $\tilde{W}$ satisfies the constraints. The implementability
condition (\ref{eq:imp-cond}) implies that $u_{1}+\nabla_{k}-\nabla_{1}\in\operatorname{Col}\mathcal{E}_{P}G$.
We can then write $u_{1}+\nabla_{k}-\nabla_{1}=\mathcal{E}_{P}Gv$
for some vector $v$. We now check the constraints. 
\begin{align*}
\mathcal{E}_{P}\tilde{W}_{k} & =\mathcal{E}_{P}G\tilde{W}_{k}'+\mathcal{E}_{P}\left(GG^{-}-I\right)\mathcal{E}_{P}^{-}\mathcal{E}_{P}Gv\\
 & =\bm{0}+\mathcal{E}_{P}GG^{-}\mathcal{E}_{P}^{-}\mathcal{E}_{P}Gv-\mathcal{E}_{P}\mathcal{E}_{P}^{-}\mathcal{E}_{P}Gv\\
 & =\mathcal{E}_{P}Gv-\mathcal{E}_{P}Gv\\
 & =\bm{0},
\end{align*}
where the second line observes $\mathcal{E}_{P}G\tilde{W}'=\boldsymbol{0}$,
and the third line observes that $AA^{-}A=A$ for any matrix $A$
and its pseudo-inverse $A^{-}$. 
\[
\mathcal{E}_{P}^{-}\left(u_{1}+\nabla_{k}-\nabla_{1}\right)+\tilde{W}_{k}=GG^{-}\mathcal{E}_{P}^{-}\left(u_{1}+\nabla_{k}-\nabla_{1}\right)+G\tilde{W}_{k}'\geq\bm{0},\forall1\leq k\leq K
\]
follows from the presumption that $G^{-}\mathcal{E}_{P}^{-}\left(u_{1}+\nabla_{k}-\nabla_{1}\right)+\tilde{W}_{k}'\geq\bm{0},\forall1\leq k\leq K$,
and that $G\geq0$. 

\end{proof}

\begin{proof}
[Proof of Proposition \ref{prop:binary-LL-characterization}] Suppose
$N=2$ and $\mathcal{E}_{P}$ and $\mathcal{E}_{P}'$ are invertible
with $\mathcal{E}_{P}'=\mathcal{E}_{P}G$ for some $G\in\mathbb{R}^{2\times2}$.
Invertibility implies that $G\bm{1}=\bm{1}$. To see this, from $\mathcal{E}_{P}'=\mathcal{E}_{P}G$,
we have $\mathcal{E}_{P}G\bm{1}=\mathcal{E}_{P}'\bm{1}=\bm{1}$. Since
$\mathcal{E}_{P}$ is invertible, and $\mathcal{E}_{P}\bm{1}=\bm{1},$we
must have $G\bm{1}=\bm{1}$. 

By Proposition \ref{prop:LL-opt-contract}, $\mathcal{E}_{P}\geq_{\text{EPIR}}\mathcal{E}_{P}'$
if and only if
\begin{equation}
\mathcal{E}_{P}\min\mathcal{E}_{P}^{-}\nabla-\mathcal{E}_{P}^{\prime}\min\mathcal{E}_{P}^{\prime-}\nabla\geq\bm{0},\label{eq:EPIR-cost-inequality-without-mu0}
\end{equation}
for any $\nabla\in\mathcal{V}$, where $\mathcal{V}$ is the set of
(conformable) matrices $\nabla$ that satisfy the no dominance condition
in Lemma \ref{lem:nabla-no-dominance}. It suffices to show that (\ref{eq:EPIR-cost-inequality-without-mu0})
is equivalent to the condition in Proposition \ref{prop:binary-LL-characterization}.

I now simplify the problem. First, since the state space is binary,
Lemma \ref{lem:V-col-reduction} implies that it suffices to focus
on $\nabla\in\mathcal{V}\cup\mathbb{R}^{2\times2}$ in (\ref{eq:EPIR-cost-inequality-without-mu0}).
That is, it suffices to consider experiments $\mathcal{E}_{A}$'s
with two realizations. I will provide the formal argument for it
after the main proof. 

Since $\mathcal{E}_{P}$ and $\mathcal{E}_{P}'$ are binary-binary,
we can without loss assume 
\[
\mathcal{E}_{P}=\begin{bmatrix}a & 1-a\\
1-b & b
\end{bmatrix},\;\mathcal{E}_{P}'=\begin{bmatrix}a' & 1-a'\\
1-b' & b'
\end{bmatrix}
\]
with $a+b\geq1$ and $a'+b'\geq1$.\footnote{This can always be obtained by reordering the columns.}
I can rewrite (\ref{eq:EPIR-cost-inequality-without-mu0}) using the
functional form. Specifically, due to linearity, it suffices to consider
\[
\nabla=\begin{bmatrix}1 & 0\\
0 & 0
\end{bmatrix}\text{ and }\begin{bmatrix}0 & 0\\
0 & 1
\end{bmatrix}.
\]
 (\ref{eq:EPIR-cost-inequality-without-mu0}) then becomes the following
\begin{align*}
\ell_{2}-\ell_{1} & \geq\ell_{2}'-\ell_{1}',\\
\dfrac{1}{\ell_{1}}-\dfrac{1}{\ell_{2}} & \geq\dfrac{1}{\ell_{1}'}-\dfrac{1}{\ell_{2}'},\\
\ell_{1}/\ell_{2} & \leq\ell_{1}'/\ell_{2}',
\end{align*}
where the first two inequalites are the conditions we need, and the
last inequality is implied by the first two.

\end{proof}
\begin{lem}
\label{lem:V-col-reduction} Suppose $\mathcal{E}_{P}\in E^{M}$ and
$\mathcal{E}_{P}'\in E^{M'}$ have full row rank with $\mathcal{E}_{P}'=\mathcal{E}_{P}G$
for some $G\in\mathbb{R}^{M\times M'}$.\footnote{The existence of such $G$ is guaranteed by the full-row-rankness
of $\mathcal{E}_{P}$.} Let $\mathcal{V}$ be the set of conformable matrices $V$ such that
$\nabla=\mathcal{E}_{P}GV$ satisfies the no dominance condition.
(\ref{eq:EPIR-cost-inequality-without-mu0}) holds for any $V\in\mathcal{V}$
if and only if it holds for any $V\in\mathcal{V}\cup\mathbb{R}^{M'\times M'}$.
\end{lem}
\begin{proof}
The only if direction is obvious. For the if direction, suppose to
the contrary that (\ref{eq:EPIR-cost-inequality}) holds for any $V\in\mathcal{V}\cap\mathbb{R}^{M'\times M'}$,
but there exists some $\tilde{V}\in\mathcal{V}\cup\mathbb{R}^{M'\times\tilde{M}}$
with $\tilde{M}>M'$ that violates (\ref{eq:EPIR-cost-inequality}).
I show that this generates a contradiction because I can construct
some $V\in\mathcal{V}\cap\mathbb{R}^{M'\times M'}$ that violates
(\ref{eq:EPIR-cost-inequality}). Since $\tilde{V}$ violates (\ref{eq:EPIR-cost-inequality}),
we must have $\mathcal{E}_{P}\min G\tilde{V}-\mathcal{E}_{P}G\min\tilde{V}<0$
in at least one component. Since $G\tilde{V}$ has $\tilde{M}>M'$
columns but only $M'$ rows, there must be at least one column $\tilde{V}_{k}$
in $\tilde{V}$ that does not affect $\min G\tilde{V}$ in the sense
that removing this column does not affect the value of $\min G\tilde{V}$.\footnote{This holds even if there are ties. Suppose both $\tilde{V}_{1}$ and
$\tilde{V}_{2}$ minimizes $G^{m}\tilde{V}_{k}$ over $k$. Removing
either $\tilde{V}_{1}$ or $\tilde{V}_{2}$ from $\tilde{V}$ does
not change this minimum.} Remove this column from $\tilde{V}$. By construction, $\min G\tilde{V}$
is not affected. On the other hand, this must weakly raise the value
of $\min\tilde{V}$ since the set to take the minimum over shrinks.
This means that $V$ defined as $\tilde{V}$ without $\tilde{V}_{k}$
also violates (\ref{eq:EPIR-cost-inequality}) for some $\mu_{0}$.
We can keep doing this as long as $\tilde{V}$ has more columns than
rows. This iterative process must end up with some $V$ with no more
than $M'$ columns and it violates (\ref{eq:EPIR-cost-inequality})
for some $\mu_{0}$. This completes the proof.
\end{proof}

\section{Information Cost Functions \label{appsec:smooth-costs}}

This appendix provides details on the agent's information cost functions.
Appendix Section \ref{appsubsec:post-sep-details} supplements Section
\ref{subsec:info-cost} by formally specifying the differentiability
assumption for posterior separable cost functions. Appendix Section
\ref{appsubsec:smooth-costs} extends the model beyond the posterior
separable costs and presents the required assumptions for a more general
class of smooth functions. Appendix Section \ref{appsubsec:smooth-costs-results}
then shows that these assumptions lead to the same set of results
in the main text. 

\subsection{Differentiability under Posterior Separability \label{appsubsec:post-sep-details}}

In the main text, I focus on the case where the agent's information
cost is posterior separable and satisfies Assumptions \ref{assu:tech-post-sep}-\ref{assu:inf-slope-post-sep}.
Here I provide more details on Assumption \ref{assu:smooth-post-sep}
about the differentiability of the cost function. 

The purpose of Assumption \ref{assu:smooth-post-sep} is to define
the marginal cost of information at any posterior belief. In the case
of posterior separable costs, this requires $c$ to admit a (Gateaux)
derivative. 
\begin{defn}
[Gateaux Derivative]\label{def:MC-post-sep} Define the directional
derivative of $c:\Delta\Omega\to\bar{\mathbb{R}}:=\mathbb{R}\cup\{\pm\infty\}$
at some $\mu\in\Delta\Omega$ in the direction of $\mu'\in\Delta\Omega$
by 
\[
d_{+}c\left(\mu;\mu'\right):=\lim_{\epsilon\to0^{+}}\dfrac{1}{\epsilon}\left[c\left(\mu+\epsilon(\mu'-\mu)\right)-c\left(\mu\right)\right].
\]
Say that $\nabla_{c}:\Delta\Omega\to\bar{\mathbb{R}}^{N}$ is a (Gateaux)
derivative of $c$ if, at every $\mu\in\Delta\Omega$, $\mu\cdot\nabla_{c}(\mu)=c(\mu)$
and 
\begin{equation}
d_{+}c\left(\mu;\mu'\right)=(\mu'-\mu)\cdot\nabla_{c}(\mu)\label{eq:gateaux-def}
\end{equation}
in every direction $\mu'\in\Delta\Omega$.
\end{defn}
This is the usual definition for derivatives of functions in $\mathbb{R}^{N}$
but with a normalization convention to simplify notations. If $\nabla_{c}(\mu)$
is a derivative of $c$ at $\mu$, then so is $\nabla_{c}(\mu)+k$
for any constant $k\in\mathbb{R}$. Algebraically, this is because
the constant $k$ cancels out in Equation (\ref{eq:gateaux-def}).
In other words, the derivative only encodes the slope of the supporting
hyperplane to $c$ at $\mu$, not its intercept. 

To also encode the intercept of the supporting hyperplane using this
additional degree of freedom, I impose the normalization condition
$\mu\cdot\nabla_{c}(\mu)=c(\mu)$ which implicitly defines the constant
$k$, following \citet{lipnowski2022predicting}. This way, the supporting
hyperplane to $c$ at $\mu$ can be compactly written as,
\[
\left\{ \left(\mu',z\right):\mu'\in\Delta\Omega,z\in\mathbb{R},\mu'\cdot\nabla_{c}(\mu)=z\right\} .
\]
Geometrically, as shown in Figure \ref{fig:marginal-cost}, the marginal
cost $\nabla_{c}(\tilde{\mu})$ at some $\tilde{\mu}$ is given by
the values of the supporting hyperplane of $c$ at $\tilde{\mu}$
when evaluated at each state. 

\begin{figure}[th]
\begin{centering}
\caption{Marginal Cost of Information \label{fig:marginal-cost}}
\par\end{centering}
\begin{centering}
\medskip{}
\par\end{centering}
\begin{centering}
\begin{minipage}[t]{0.5\columnwidth}%
\begin{center}
\includegraphics[width=1\columnwidth]{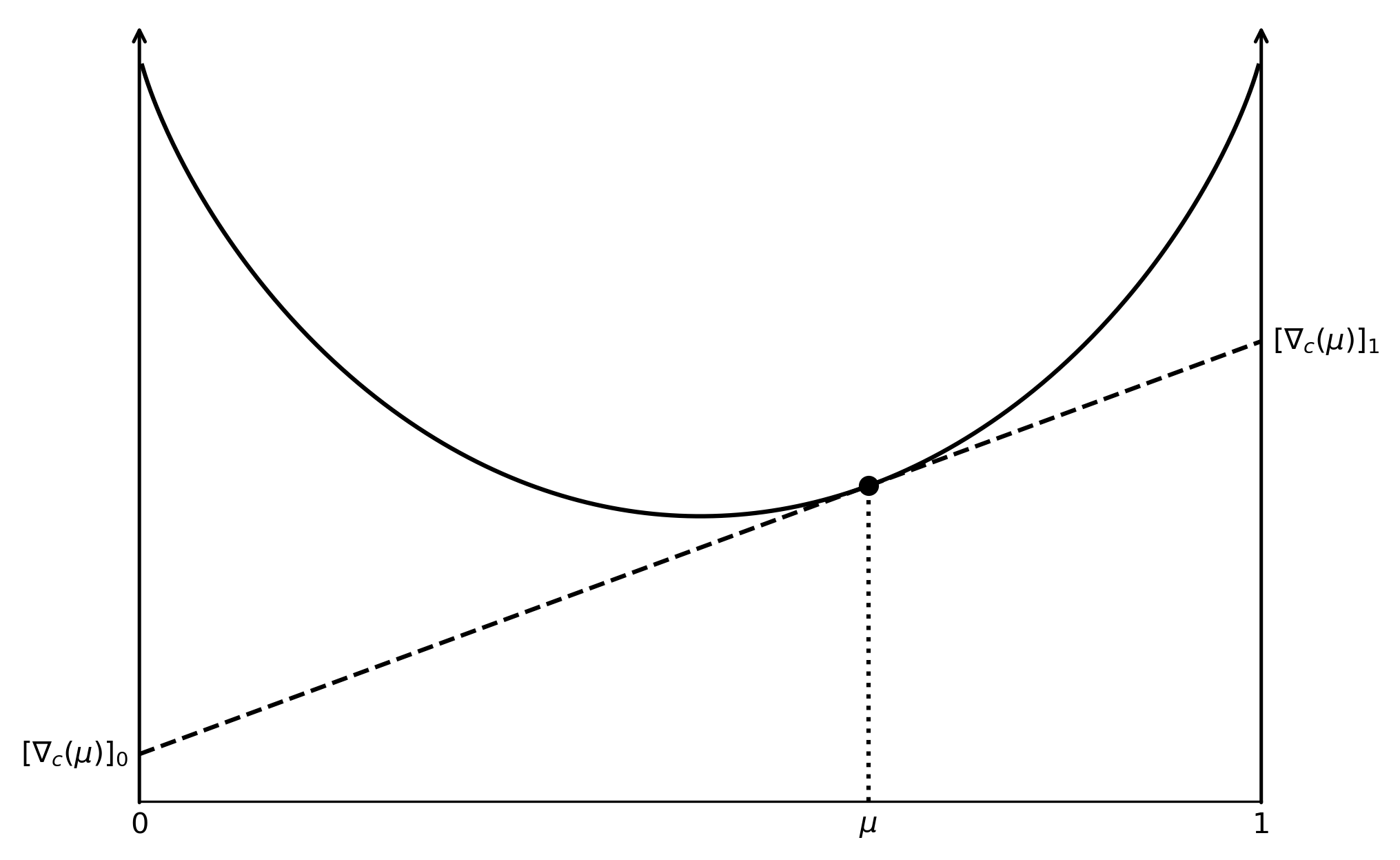}
\par\end{center}%
\end{minipage}
\par\end{centering}
\centering{}%
\begin{minipage}[t]{0.9\columnwidth}%
\begin{singlespace}
{\scriptsize Notes: This figure illustrates the cost and the marginal
cost of information when the state space is binary $\Omega=\{\omega_{0},\omega_{1}\}$.
The horizontal axis is the belief $\mu\in[0,1]$ representing the
probability of $\omega_{1}$. The vertical axis is the cost. The solid
convex curve is the posterior cost function $c$. The dashed line
is the supporting hyperplane of $c$ at belief $\tilde{\mu}$. The
marginal cost $\nabla_{c}(\tilde{\mu})=\left(\left[\nabla_{c}(\tilde{\mu})\right]_{0},\left[\nabla_{c}(\tilde{\mu})\right]_{1}\right)$,
under the normalization introduced above, is given by the values of
this supporting hyperplane at $\mu=0$ and $\mu=1$. }{\scriptsize\par}
\end{singlespace}

\end{minipage}
\end{figure}

\subsection{General Cost Functions \label{appsubsec:smooth-costs}}

The information cost function $C:\Delta\Delta\Omega\to\mathbb{R}_{+}\cup\{+\infty\}$
is assumed to satisfy Assumptions \ref{assu:tech}-\ref{assu:inf-slope}.
They generalize Assumptions \ref{assu:tech-post-sep}-\ref{assu:inf-slope-post-sep}
in the main text.

\begin{assumption}
\label{assu:tech}The agent’s information cost $C$ is 
\begin{enumerate}
\item Lower semi-continuous;
\item Convex;
\item Zero cost at no information: $C(\underline{\mathcal{E}})=0$ where
$\underline{\mathcal{E}}$ is an uninformative experiment;
\item (Blackwell) monotone: $\mathcal{E}\geq_{B}\mathcal{E}'\Rightarrow C(\mathcal{E})\geq C(\mathcal{E}')$.\footnote{Abusing the notation, I use $\mathcal{E}$ and $\mathcal{E}'$ to
also denote the distributions of posteriors they induce.}
\end{enumerate}
\end{assumption}
Assumption \ref{assu:tech} generalizes Assumption \ref{assu:tech-post-sep}.
To see the connections, lower semi-continuity and monotonicity of
$C$, as well as $C(\underline{\mathcal{E}})=0$ follow directly from
Assumption \ref{assu:tech-post-sep} when $C$ is posterior separable.
Convexity is an additional requirement that says that the agent can
weakly save costs by randomization. It keeps the agent's problem well-behaved
and is automatically satisfied when $C$ is posterior separable. We
do not have to assume convexity when $C$ is posterior separable because
a posterior separable cost is affine in the posterior distribution
and therefore always convex by construction. It is also worth noting
that assuming convexity and monotonicity of $C$ is without loss \citep{caplin2015revealed}.

Smoothness requires additional work. Following \citet{lipnowski2022predicting},
I assume the cost function $C$ is differentiable with a differentiable
derivative. This again allows me to define the marginal cost of information
at any posterior belief induced by an experiment. To formally state
this assumption, recall the following definitions. 
\begin{defn}
\label{def:general-derivatives}Define the (directional) derivative
of $C:\Delta\Delta\Omega\to\mathbb{R}_{+}\cup\{+\infty\}$ at some
$\mathcal{E}\in\Delta\Delta\Omega$ in the direction of $\mathcal{E}'\in\Delta\Delta\Omega$
by
\[
d_{+}C(\mathcal{E};\mathcal{E}'):=\lim_{\epsilon\to0^{+}}\dfrac{1}{\epsilon}\left[C\left(\mathcal{E}+\epsilon(\mathcal{E}'-\mathcal{E})\right)-C\left(\mathcal{E}\right)\right].
\]
Say that $c_{\mathcal{E}}:\Delta\Omega\to\bar{\mathbb{R}}$ is a
(Gateaux) derivative of $C$ at some $\mathcal{E}\in\Delta\Delta\Omega$
if $\int c_{\mathcal{E}}(\mu)\mathcal{E}(d\mu)$ is finite and for
every $\mathcal{E}'$ with $C(\mathcal{E}')<+\infty$, 
\[
d_{+}C(\mathcal{E};\mathcal{E}')=\int_{\Delta\Omega}c_{\mathcal{E}}(\mu)\left(\mathcal{E}'-\mathcal{E}\right)(d\mu).
\]
Say that $C$ is differentiable at $\mathcal{E}\in\Delta\Delta\Omega$
if it admits a derivative $c_{\mathcal{E}}$.

Furthermore, define the derivative of $c_{\mathcal{E}}:\Delta\Omega\to\bar{\mathbb{R}}$
at some $\mu\in\Delta\Omega$ in the direction of $\mu'\in\Delta\Omega$
by 
\[
d_{+}c_{\mathcal{E}}(\mu;\mu'):=\lim_{\epsilon\to0^{+}}\dfrac{1}{\epsilon}\left[c_{\mathcal{E}}\left(\mu+\epsilon(\mu'-\mu)\right)-c_{\mathcal{E}}\left(\mu\right)\right].
\]
Say that $\nabla_{c_{\mathcal{E}}}:\Delta\Omega\to\bar{\mathbb{R}}^{N}$
is a (Gateaux) derivative of $c_{\mathcal{E}}$ if, at every $\mu\in\operatorname{Supp}\mathcal{E}$,
$\mu\cdot\nabla_{c_{\mathcal{E}}}(\mu)=c_{\mathcal{E}}(\mu)$ and
\[
d_{+}c_{\mathcal{E}}(\mu;\mu')=(\mu'-\mu)\cdot\nabla_{c_{\mathcal{E}}}(\mu)
\]
in every direction $\mu'\in\Delta\Omega$.\footnote{The requirement of $\mu\cdot\nabla_{c_{\mathcal{E}}}(\mu)=c_{\mathcal{E}}(\mu)$
is a normalization so that $\nabla_{c_{\mathcal{E}}}(\mu)$ represents
the subgradient of $c_{\mathcal{E}}$ at $\mu$.}
\end{defn}
\begin{assumption}
\label{assu:smooth}The agent’s information cost $C$ is smooth, that
is, $C$ is 
\begin{enumerate}
\item Differentiable: For every $\mathcal{E}$ with $C(\mathcal{E}')<+\infty$,
$C$ admits a derivative $c_{\mathcal{E}}$;
\item Iteratively differentiable: For every $\mathcal{E}$ with $C(\mathcal{E}')<+\infty$,
the derivative $c_{\mathcal{E}}$ again admits a derivative $\nabla_{c_{\mathcal{E}}}$.
\end{enumerate}
\end{assumption}
Compared to Assumption \ref{assu:smooth-post-sep}, Assumption \ref{assu:smooth}.1
requires additionally that $C$ is differentiable, otherwise its derivative
$c_{\mathcal{E}}$ is not well-defined. Assumption \ref{assu:smooth}.2
requires that the derivative $c_{\mathcal{E}}$ at any feasible $\mathcal{E}$
is again differentiable. This is exactly Assumption \ref{assu:smooth-post-sep}
with one caveat that, the iterative derivative $\nabla_{c_{\mathcal{E}}}$
now depends on both the posterior distribution $\mathcal{E}$ and
the induced posterior $\mu$. The dependence on $\mathcal{E}$ comes
from relaxing posterior separability because the price of posterior
$\mu$ now depends also on what information the agent is acquiring. 

Similar to Assumption \ref{assu:smooth-post-sep}, Assumption \ref{assu:smooth}
allows me to define the marginal cost of learning $\nabla_{c_{\mathcal{E}}}$,
which now depends on both $\mathcal{E}$ and $\mu$. To see this connection,
start from a differentiable cost $C$. From Definition \ref{def:general-derivatives},
differentiability says that $C$ prices marginal changes in information
in a posterior-separable manner. It can then be locally approximated
around any feasible $\mathcal{E}$ by a posterior-separable cost $C_{c}(\mathcal{E}):=\int c_{\mathcal{E}}(\mu)\mathcal{E}(d\mu)$.
The interpretation of $c_{\mathcal{E}}(\mu)$ is then the approximate
price of the posterior belief $\mu$. As a result, the same intuition
from the posterior separable costs also follows here. Consider a marginal
change to $\mathcal{E}$ by changing some posterior $\mu\in\operatorname{Supp}\mathcal{E}$.
The iterative derivative $\nabla_{\mu}c_{\mathcal{E}}$ then describes
how this approximate price changes when the posterior $\mu$ changes.

Lastly, sometimes I assume that it is infinitely costly to rule out
any state. This is a restatement of \ref{assu:inf-slope-post-sep}
without posterior separability.
\begin{assumption}
\label{assu:inf-slope}The agent’s information cost $C$ has infinite
slopes at the boundary: If $\mathcal{E}$ induces any non-fully-supported
posterior, then 
\[
d_{+}C(\mathcal{E};\underline{\mathcal{E}})=-\infty.
\]
\end{assumption}

\subsection{Results under General Cost Functions \label{appsubsec:smooth-costs-results}}

I now generalize the main results to general smooth cost function
under the assumptions in Appendix \ref{appsubsec:smooth-costs}. All
the results remain valid under general cost functions when we replace
Assumptions \ref{assu:tech-post-sep}-\ref{assu:inf-slope-post-sep}
with Assumptions \ref{assu:tech}-\ref{assu:inf-slope}, respectively,
and change the marginal cost vectors from $\nabla_{c}$ to $\nabla_{c_{\mathcal{E}}}$
as they now also depend on $\mathcal{E}$. 

The only result that requires additional proof is Lemma \ref{lem:imp-matrix-eqn},
as the first order condition becomes less obvious without the posterior
separable form. To that end, I follow Lemma 1 of \citet{georgiadis2024flexible}. 
\begin{proof}
[Proof of Lemma  \ref{lem:imp-matrix-eqn} for General Smooth Costs]
The agent's problem is given by \ref{eq:A-problem} and does not have
the posterior separable form. Due to differentiability, around any
feasible $\mathcal{E}$, $C$ can be approximated by its local affine
approximation, which has a posterior separable form,
\[
C_{c}(\mathcal{E}'):=\int c_{\mathcal{E}}(\mu)\mathcal{E}'(d\mu).
\]
To be more specific, suppose $c$ is a derivative of $C$ at some
feasible $\mathcal{E}^{*}$ and $C_{c}(\cdot)$ is the local affine
approximation of $C$ around $\mathcal{E}^{*}$, then we have, 
\[
\mathcal{E}^{*}\in\underset{\mathcal{E}\in E^{K}}{\arg\max}\:\mathbb{E}_{\mu_{k}\sim\mathcal{E}}\left[\mu_{k}\cdot\mathcal{E}_{P}t_{k}\right]-C(\mathcal{E})
\]
if and only if 
\[
\mathcal{E}^{*}\in\underset{\mathcal{E}\in E^{K}}{\arg\max}\:\mathbb{E}_{\mu_{k}\sim\mathcal{E}}\left[\mu_{k}\cdot\mathcal{E}_{P}t_{k}\right]-C_{c}(\mathcal{E}).
\]
This is indeed Lemma 1 of \citet{georgiadis2024flexible}. Once the
above equivalence is established, the remainder of the proof follows
exactly as before, with $c_{\mathcal{E}}$ replacing $c$. 
\end{proof}
After proving Lemma \ref{lem:imp-matrix-eqn} in the main text for
general smooth costs, Theorems \ref{thm:imp} and \ref{thm:imp-contract-decomposition}
follow immediately, with Assumptions \ref{assu:tech-post-sep}-\ref{assu:inf-slope-post-sep}
replaced by Assumptions \ref{assu:tech}-\ref{assu:inf-slope}, and
the marginal cost vectors redefined as $\nabla_{k}:=\nabla_{c_{\mathcal{E}}}(x_{k})$.
All other results then follow from that.

\section{Implementability with Corner Solutions \label{appsec:imp-corner-solutions}}

In this appendix, I present the implementability result without Assumption
\ref{assu:inf-slope-post-sep}. I first discuss what will happen without
Assumption \ref{assu:inf-slope-post-sep}. In this case, nothing changes
if the principal wants to implement some $\mathcal{E}_{A}$ that never
rules out any state. The FOCs in Equation (\ref{eq:imp-matrix-eqn})
must still hold as equalities due to the interiority of beliefs. However,
when $\mathcal{E}_{A}$ has at least one non-fully-supported posterior,
the characterization will be more complicated. Intuitively, the principal
may provide enough incentives for the agent to rule out a state following
some realization(s), which is not possible under Assumption \ref{assu:inf-slope-post-sep}.
Under such incentives, the agent can at most lower the posterior belief
of this state to zero but not any further. As a result, the FOCs may
hold as an inequality and additional multipliers are required, as
summarized in Theorem \ref{thm:imp-corner}. The only difference from
Theorem \ref{thm:imp} is the introduction of the additional multipliers
$\eta_{k}$ on the constraint that the probabilities must be positive. 
\begin{thm}
[Characterization of Implementability, Potential Corner Solutions]\label{thm:imp-corner}
Suppose $c$ satisfies Assumptions \ref{assu:tech-post-sep}-\ref{assu:smooth-post-sep}.
An experiment $\mathcal{E}_{A}\in E^{K}$ with posteriors $\left\{ x_{k}\right\} _{k=1}^{K}$
is implementable under $\mathcal{E}_{P}$ if and only if $C(\mathcal{E}_{A})<+\infty$
and $\exists\eta_{k}\in\mathbb{R}_{+}^{N}$ for $1\leq k\leq K$ such
that 
\begin{equation}
\nabla_{k}+\eta_{k}-\nabla_{k'}-\eta_{k'}\in\operatorname{Col}\mathcal{E}_{P},\forall1\leq k,k'\leq K,\label{eq:imp-cond-corner}
\end{equation}
with 
\[
\gamma_{k}^{n}x_{k}^{n}=0,\forall1\leq k\leq K,1\leq n\leq N,
\]
where $\nabla_{k}:=\nabla_{x_{k}}c$ is the marginal cost vector at
posterior $x_{k}$.
\end{thm}
\begin{proof}
Following the proof of Theorem \ref{thm:imp}, we can rewrite the
agent's problem as choosing a distribution of posteriors to maximize
(\ref{eq:A-problem-post-sep}) subject to the Bayes plausibility constraint
(\ref{eq:Bayes-plausibility}) and the constraint that probabilities
must be positive 
\begin{equation}
\mu_{k}\geq0,\forall1\leq k\leq K.\label{eq:prob-positive}
\end{equation}
In Theorem \ref{thm:imp}, we do not have to worry about (\ref{eq:prob-positive})
since non-interior posteriors will result in an infinite marginal
cost due to Assumption \ref{assu:inf-slope-post-sep} and such experiments
can never be implementable. Without Assumption \ref{assu:inf-slope-post-sep},
we have to take care of (\ref{eq:prob-positive}). Let $\eta_{k}\in\mathbb{R}_{+}^{N}$
be the multipliers on (\ref{eq:prob-positive}). Using the same techniques,
the first order conditions in the matrix form is, 
\begin{equation}
\mathcal{E}_{P}T=\nabla+\lambda\bm{1}_{1\times K}+\eta,\label{eq:imp-matrix-eqn-corner}
\end{equation}
where $\eta=\begin{bmatrix}\eta_{1} & \eta_{2} & \cdots & \eta_{K}\end{bmatrix}$
is the matrix of the additional multipliers. Using the same linear
algebra trick as in the proof of Theorem \ref{thm:imp}, (\ref{eq:imp-matrix-eqn-corner})
is equivalent to Condition (\ref{eq:imp-cond-corner}).
\end{proof}

\section{Unique Implementability \label{appsec:imp-strict}}

This appendix extends the analysis in Section \ref{sec:implementability}
and studies the question of unique implementability, that is, whether
the principal can design a contract so that the target experiment
is the agent's unique optimal choice. 
\begin{defn}
[Unique Implementability] Say that a contract $T$ under $\mathcal{E}_{P}$
uniquely implements $\mathcal{E}_{A}$ if $\mathcal{E}_{A}$ is the
unique optimal solution to the agent's problem (\ref{eq:agent-IC}).
Say that $\mathcal{E}_{A}$ is uniquely implementable under $\mathcal{E}_{P}$
if there exists some contract $T$ under $\mathcal{E}_{P}$ that uniquely
implements $\mathcal{E}_{A}$. 
\end{defn}
Unique implementability poses two additional requirements. The first
is the strict convexity of the information cost function, otherwise
the agent can perturb the experiment within the affine portion of
the cost function and maintain optimality. Define strict convexity
around a point as follows. 
\begin{defn}
[Strict Convexity Around a Point] Say that a convex function $f:X\to\bar{\mathbb{R}}$
is strictly convex around some $x\in X$ if there exists an open neighborhood
$U\subseteq X$ containing $x$ such that for all $x',x''\in U$ with
$x'\neq x''$ and all $\alpha\in(0,1)$, 
\[
f\left(\alpha x'+(1-\alpha)x''\right)<\alpha f(x')+(1-\alpha)f(x'').
\]
\end{defn}
I need the information cost function $C$ to be strictly convex around
$\mathcal{E}_{A}$. If $C$ is posterior separable with posterior
cost $c$, this reduces to the strict convexity of $c$ around the
posteriors induced by $\mathcal{E}_{A}$.

Second, I need the induced posteriors to be linearly independent.
Intuitively, this is because the posterior separable cost is affine
in the distribution of posteriors. Therefore, the first order condition
(\ref{eq:imp-matrix-eqn}) only pins down the marginal cost at every
posterior induced by $\mathcal{E}_{A}$ without restricting the probabilities
of the posteriors. With linearly dependent posteriors, there can be
multiple Bayes plausible posterior distributions satisfying the first
order condition, breaking the uniqueness of $\mathcal{E}_{A}$. This
argument holds even if the information cost is not posterior separable,
as long as it is smooth in the way defined in Appendix Section (\ref{appsubsec:smooth-costs}).
This is because smooth cost functions can be locally approximated
by posterior separable cost. 

The next proposition characterizes the conditions for unique implementability.
\begin{prop}
[Unique Implementability] An experiment $\mathcal{E}_{A}\in E^{K}$
with posteriors $\left\{ x_{k}\right\} _{k=1}^{K}$ is uniquely implementable
if and only if (1) it is implementable, (2) $c$ is strictly convex
around every $x_{k}$, and (3) the induced posteriors $x_{k}$ are
linearly independent.\footnote{If the information cost function $C$ is not posterior separable,
simply change (2) to ``$C$ is strictly convex around $\mathcal{E}_{A}$''.}
\end{prop}
\begin{proof}
For the ``if'' direction, suppose (1)-(3) hold. (1) implies that
there exists some contract $T$ so that the first order condition
(\ref{eq:imp-matrix-eqn}) can be satisfied at posteriors $\left\{ x_{k}\right\} _{k=1}^{K}$.
(3) says that there is a unique distribution over $\left\{ x_{k}\right\} _{k=1}^{K}$
that averages back to the prior, and this has to be $\mathcal{E}_{A}$.
(2) says that the agent's objective in (\ref{eq:agent-IC}) is strictly
convex. Therefore, $\mathcal{E}_{A}$ is the unique solution to the
agent's problem. 

For the ``only if'' direction, suppose $\mathcal{E}_{A}$ is uniquely
implementable. This implies (1) since $\mathcal{E}_{A}$ must be implementable.
Let $T$ be the contract that uniquely implements $\mathcal{E}_{A}$.
Suppose (2) does not hold. There then exists another experiment $\mathcal{E}_{A}'$
that also satisfies the first order condition. One way to construct
this is to take some $x_{k}$ induced by $\mathcal{E}_{A}$ but $c$
is not strictly convex around it. There exists some open neighborhood
$U$ containing $x_{k}$ and $c$ is affine over $U$. We can then
split $x_{k}$ into some $x_{k}',x_{k}''\in U$ with $x_{k}'=x_{k}+\epsilon$
and $x_{k}''=x_{k}-\epsilon$ for some $\epsilon\in\mathbb{R}^{N}$
with $\epsilon\cdot\boldsymbol{1}=0$ and $\left\Vert \epsilon\right\Vert $
small enough. Define $\mathcal{E}_{A}'$ as the experiment that is
the same as $\mathcal{E}_{A}$, except that $\mathcal{E}_{A}'$ indices
$x_{k}'$ and $x_{k}''$ with equal probability when $\mathcal{E}_{A}$
induces $x_{k}$. Since $c$ is affine over $U$, the marginal costs
of learning at $x_{k}'$ and $x_{k}''$ are the same as that at $x_{k}$.
By reporting $x_{k}$ whenever $x_{k}'$ or $x_{k}''$ realize, $\mathcal{E}_{A}'$
also satisfies the first order condition and is hence optimal, breaking
the unique optimality of $\mathcal{E}_{A}$. Suppose (3) does not
hold. There are multiple convex combinations of $\left\{ x_{k}\right\} _{k=1}^{K}$
that averages back to the prior. All of them satisfy the first order
condition and are therefore optimal, again breaking the unique optimality
of $\mathcal{E}_{A}$.
\end{proof}

\section{Implementability under Deficient Row Rank\label{appsec:imp-deficient-rank}}

When $\mathcal{E}_{P}$ does not have full row rank, not all experiments
are implementable. From Theorem \ref{thm:imp}, the implementability
of certain $\mathcal{E}_{A}$ depends on both the column space of
$\mathcal{E}_{P}$ and the exact shape of the cost function at $\mathcal{E}_{A}$.
There is no simple characterization for the implementable set of experiments.
I illustrate this with the following example. 
\begin{example}
\label{exa:imp-eA} Three states $\Omega=\left\{ \omega_{1},\omega_{2},\omega_{3}\right\} $
with uniform prior $\mu_{0}$. The agent has an entropy reduction
cost $c(\mu):=H(\mu_{0})-H(\mu)$ with $H(\mu):=-\sum_{n=1}^{3}\mu^{n}\log\mu^{n}$
where $\mu^{n}$ is the probability of state $\omega_{n}$ at belief
$\mu$. The marginal cost vector at $\mu$ is $\nabla(\mu)=\begin{bmatrix}\log\mu^{1} & \log\mu^{2} & \log\mu^{3}\end{bmatrix}^{\intercal}$
up to a normalizing constant.\footnote{See Definition \ref{def:MC-post-sep}.}
Consider the following two contractible experiments with deficient
row ranks, each with only two realizations, 
\[
\mathcal{E}_{1}=\begin{bmatrix}\frac{3}{8} & \frac{5}{8}\\
\frac{3}{8} & \frac{5}{8}\\
\frac{3}{4} & \frac{1}{4}
\end{bmatrix};\mathcal{E}_{2}=\begin{bmatrix}\frac{3}{4} & \frac{1}{4}\\
\frac{1}{4} & \frac{3}{4}\\
\frac{1}{2} & \frac{1}{2}
\end{bmatrix}.
\]

Experiment $\mathcal{E}_{1}$ learns only about $\omega_{3}$ but
not about the relative likelihood of $\omega_{1}$ and $\omega_{2}$.
It induces posterior beliefs $\left(\frac{1}{4},\frac{1}{4},\frac{1}{2}\right)$
and $\left(\frac{5}{12},\frac{5}{12},\frac{1}{6}\right)$. Applying
Theorem \ref{thm:imp}, $\mathcal{E}_{A}$ is implementable if and
only if any induced posteriors $x_{k}$ and $x_{k'}$ satisfies $\nabla(x_{k})-\nabla(x_{k'})\in\operatorname{Col}\mathcal{E}_{P}$,
which says $\mathcal{E}_{A}$ must be supported on posteriors $x$'s
such that $\frac{x^{1}}{x^{2}}=\text{Constant}$, where $x^{n}$
is the probability of state $\omega_{n}$ at posterior $x$. Bayes
plausibility requires $\mu_{0}\in\operatorname{co}\left(\operatorname{Supp}\mathcal{E}_{A}\right)$,
which dictates that the constant must be 1. Figure \ref{fig:imp-example}
plots the set of such posteriors. Implementable $\mathcal{E}_{A}$'s
can only induce posteriors on the blue line segment which is the affine
span of the posteriors induced by $\mathcal{E}_{1}$. 

\begin{figure}[t]
\begin{centering}
\caption{Implementable $\mathcal{E}_{A}$'s in Example \ref{exa:imp-eA} \label{fig:imp-example}}
\par\end{centering}
\begin{centering}
\bigskip{}
\par\end{centering}
\begin{centering}
\begin{minipage}[t]{0.45\columnwidth}%
\begin{center}
{\footnotesize Panel A: Implementability under $\mathcal{E}_{1}$}{\footnotesize\par}
\par\end{center}
\begin{center}
\includegraphics[width=1\columnwidth]{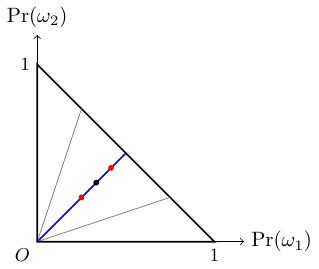}
\par\end{center}%
\end{minipage}\quad{}%
\begin{minipage}[t]{0.45\columnwidth}%
\begin{center}
{\footnotesize Panel B: Implementability under $\mathcal{E}_{2}$}{\footnotesize\par}
\par\end{center}
\begin{center}
\includegraphics[width=1\columnwidth]{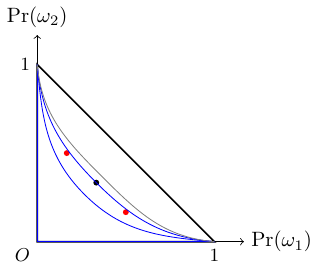}
\par\end{center}%
\end{minipage}
\par\end{centering}
\centering{}%
\begin{minipage}[t]{0.9\columnwidth}%
\begin{singlespace}
{\scriptsize Notes: This figure depicts the set of posteriors induced
by implementable $\mathcal{E}_{A}$'s. Panels A and B present implementability
under $\mathcal{E}_{1}$ and $\mathcal{E}_{2}$ in Example \ref{exa:imp-eA},
respectively. In each panel, the two axes are the probabilities of
$\omega_{1}$ and $\omega_{2}$, and the residual probability is for
$\omega_{3}$. The black dot represents the prior belief. The red
dots are the posterior beliefs induced by either $\mathcal{E}_{1}$
or $\mathcal{E}_{2}$. In Panel A, the blue line segment passing through
the origin and the prior is the set of posteriors that can be induced
by any implementable $\mathcal{E}_{A}$. Other thin grey line segments
are posteriors that also satisfy $x^{1}/x^{2}=\text{Constant}$ but
are ruled out by Bayes plausibility. In Panel B, the blue elliptic
arcs are defined by $x^{1}x^{2}/\left(1-x^{1}-x^{2}\right)^{2}=\text{Constant}$
for any constant no larger than one. There is a continuum of such
arcs, but any implementable $\mathcal{E}_{A}$'s must be supported
on the same elliptic arc. The thin grey arcs are defined by a larger-than-one
constant. Posterior distributions supported on them are clearly not
Bayes plausible.}{\scriptsize\par}
\end{singlespace}

\end{minipage}
\end{figure}

However, it is not in general true that the implementable $\mathcal{E}_{A}$'s
induce posteriors in the affine span of the posteriors induced by
the principal's experiment. $\mathcal{E}_{2}$ is such an example.
Experiment $\mathcal{E}_{2}$ learns only the relative likelihood
of $\omega_{1}$ and $\omega_{2}$ but nothing about $\omega_{3}$.
It induces posterior beliefs $\left(\frac{1}{2},\frac{1}{6},\frac{1}{3}\right)$
and $\left(\frac{1}{6},\frac{1}{2},\frac{1}{3}\right)$. Following
the same steps, an implementable $\mathcal{E}_{A}$ must be supported
on posteriors $x$'s such that $\frac{x^{1}x^{2}}{(1-x^{1}-x^{2})^{2}}=\text{Constant}$.
Given a constant, this defines an elliptic arc that $\mathcal{E}_{A}$
must be supported upon. Bayes plausibility implies that the constant
must not exceed one. More specifically, there is an additional degree
of freedom to choose the value of this constant. But any implementable
$\mathcal{E}_{A}$ must induce posteriors that are all on the same
elliptic arc. The implementable $\mathcal{E}_{A}$'s now induce posteriors
outside the affine hull of the posteriors induced by $\mathcal{E}_{2}$. 
\end{example}

\end{document}